\documentclass[conference]{IEEEtran}
\usepackage{graphicx,url}
\usepackage[sort,nospace]{cite}
\usepackage[caption=false,font=footnotesize,hangindent=10pt]{subfig}
\usepackage{amsmath,amsfonts,amsthm,bm}
\usepackage{amssymb}
\usepackage[linesnumbered,lined,boxed,ruled]{algorithm2e}
\usepackage{color,comment,xparse}
\usepackage{flushend}
\usepackage[bookmarks=false,hidelinks]{hyperref}
\usepackage{threeparttable,tabularx,multirow}
\usepackage{tikz}
\usetikzlibrary{shapes.geometric,positioning,backgrounds,fit}

\newcommand{\minitab}[1]{\begin{tabular}{@{}c@{}}#1\end{tabular}}

\graphicspath{{./}{figure/}}

\NewDocumentCommand{\E}{o m}{
\mathbb{E}\IfValueTF{#1}{_{#1}}{}\left[#2\right]
}
\newcommand{\indr}[1]{\mathbf{1}\left\{ #1 \right\}}

\newcommand{\CU}{\mathcal{U}}
\newcommand{\CV}{\mathcal{V}}
\newcommand{\CE}{\mathcal{E}}
\newcommand{\CS}{\mathcal{S}}
\newcommand{\CL}{\mathcal{L}}

\newcommand{\VS}[1]{\text{VS\textsuperscript{#1}}}
\newcommand{\RW}[1]{\text{RW\textsuperscript{#1}}}

\newtheorem{theorem}{Theorem}

\newtheorem{example}{Example}

\title{Design of Efficient Sampling Methods on Hybrid Social-Affiliation
Networks\\{\large Technique Report}}
\author{
\IEEEauthorblockN{
Junzhou Zhao\IEEEauthorrefmark{1} \quad
John C.S. Lui\IEEEauthorrefmark{2} \quad
Don Towsley\IEEEauthorrefmark{3} \quad
Pinghui Wang\IEEEauthorrefmark{4} \quad
Xiaohong Guan\IEEEauthorrefmark{1}
}
\IEEEauthorblockA{
\IEEEauthorrefmark{1}Xi'an Jiaotong University, China\\
\IEEEauthorrefmark{2}The Chinese University of Hong Kong, Hong Kong\\
\IEEEauthorrefmark{3}University of Massachusetts Amherst, USA\\
\IEEEauthorrefmark{4}Huawei Noah's Ark Lab, Hong Kong\\
}
\IEEEauthorblockA{
\{jzzhao,xhguan\}@sei.xjtu.edu.cn \,
cslui@cse.cuhk.edu.cn \,
towsley@cs.umass.edu \,
wang.pinghui@huawei.com
}}

\begin{document}
\maketitle

\begin{abstract}
Graph sampling via crawling has become increasingly popular and important
in the study of measuring various characteristics of large scale complex
networks.
While powerful, it is known to be challenging when the graph is
loosely connected or disconnected which slows down the convergence of
random walks and can cause poor estimation accuracy.

In this work, we observe that the graph under study, or called \emph{target
graph}, usually does not exist in isolation.
In many situations, the target graph is related to an \emph{auxiliary
graph} and an \emph{affiliation graph}, and the target graph becomes well
connected when we view it from the perspective of these three graphs
together, or called a \emph{hybrid social-affiliation graph} in this paper.
When directly sampling the target graph is difficult or inefficient, we can
\emph{indirectly} sample it efficiently with the assistances of the other
two graphs.
We design three sampling methods on such a hybrid social-affiliation
network.
Experiments conducted on both synthetic and real datasets demonstrate
the effectiveness of our proposed methods.

\end{abstract}


\section{\textbf{Introduction}} \label{sec:intro}

Online social networks (OSNs) such as Facebook, Sina Weibo, and Twitter
have attracted researchers' much attention in recent years because of their
ever-increasing popularity and importance in our daily
lives~\cite{Newman2003,Watts2004,Mislove2007,Lazer2009,Aral2012a}.
An OSN not only provides a platform for people to connect with their
friends, but also provides an opportunity for researchers to study user
characteristics, which are valuable for applications such as marketing
decision making.
For example, Twitter users' tweeting
activities (e.g., number of tweets related to a movie) can be used to
predict movie box-office revenues~\cite{Asur2010}, and Twitter users' mood
characteristics have a relation with stock market prices~\cite{Bollen2011}.
Therefore, measuring user characteristics in OSNs is an important task.

Exactly calculating user characteristics requires the complete OSN data.
However, for third parties who do not own the data can only rely on public
APIs to crawl the OSN.
To protect user privacy, OSNs usually impose barriers to
limit third parties' large-scale crawling~\cite{Mondal2012} and restrict
the rate of requesting APIs (e.g., Sina Weibo allows a user to issue at
most 150 requests per hour~\cite{WeiboLimit}).
As a result, crawling the complete data of a large-scale OSN is practically
impossible.

To address this challenge, sampling methods are developed, i.e., a small
fraction of OSN users are sampled and used to calculate the
characteristics.
In the literature, random walk based sampling methods have become
popular~\cite{Lovasz1993,Massoulie2006,Gjoka2011}.
A random walker starts from an initial node in the OSN, and randomly
selects a neighbor to visit at the next step; this process repeats until
the sampling budget is exhausted.
The random walk sampling can generate Markov chain samples which are able
to provide unbiased estimates of graph statistics~\cite{Meyn2009}.

\smallskip\noindent\textbf{Motivation:}
If a graph has community structure, the random walk will suffer from
\emph{slow mixing}, i.e., requiring a long burning period to reach the
steady state, which results in a substantially large number of samples
so as to keep estimation accuracy.
Recent studies have found that the mixing time in several real-world
networks is much longer than expected~\cite{Mohaisen2010}.
To overcome the slow mixing problem, one effective approach is to allow the
random walker(s) to randomly jump to (or start from) different regions of
a network, such as \emph{random walk with jumps}
(RWwJ)~\cite{Avrachenkov2010,Ribeiro2012b} and \emph{Frontier sampling}
(FS)~\cite{Ribeiro2010}.
These methods explicitly or implicitly assume that {\em random vertex
sampling} is enabled.
For example, in RWwJ, the walker can randomly jump to other nodes while
walking, and the initializing step in FS relies on uniform vertex sampling.
However, random vertex sampling can be resource-intensive when the
\emph{effective} account ID space is very sparsely populated such as the
following example.

\begin{example}\label{exam:weibo}
A restaurant company wants to build a new chain store in one of two small
candidate cities in China.
A market surveyor is sent to study the consuming ability of citizens
there.
Since most citizens use the check-in service~\cite{WeiboPlace} to share
their consuming information in Weibo, the surveyor decides to use Weibo
as a platform to conduct his research.
He plans to uniformly sample two collections of Weibo users in the two
cities respectively.
It is known that every Weibo account ID consists of ten digits ranging from
``1000000000'' to the maximum\footnote{By March 25, 2014, the maximum Weibo
user ID is about ``5058913818''.}.
He generates random numbers in this range as test IDs and finds that about
$11\%$ of the test IDs are valid Weibo users. 
However, because the population sizes of the two cities are small (e.g.,
hundreds of thousands of citizens comparing to the hundreds of millions of
Weibo users), the valid users falling into the two cities has probability
as small as $0.1\%$.
\end{example}

In the above example, an effective test ID must fall into one of the two
cities, and random vertex sampling becomes extremely inefficient because
the probability that a test ID is effective equals $P(\textit{ID is
valid})\times P(\textit{ID falls into one city})$, where $P(\textit{ID is
valid})\approx 0.11$ and $P(\textit{ID falls into one city})\approx
10^{-3}$.
This results in that the surveyor needs to try $10^4$ times on average to
obtain a valid ID falling in one of the two cities. 
To make matters worse, in some OSNs such as Pinterest, account IDs are
arbitrary-length strings, which makes random vertex sampling practically
impossible.
So, {\em how can we sample vertex randomly in an OSN when random vertex
sampling is extremely inefficient or impractical at all?}

\smallskip\noindent\textbf{Present Work:}
In Example~\ref{exam:weibo}, the key problem is how to effectively sample
Weibo users within the two cities.
We notice that the check-ins shared by users often contain the venue
information, e.g., in which restaurant the user lunched, and most such OSNs
(e.g., Foursquare) provide APIs for querying venues (e.g., restaurants)
within an area of interest by specifying a rectangle region with the
bottom-left and top-right corners latitude-longitude coordinates given, or
a circle region with the center point latitude-longitude coordinate and
radius given.
This function can be used to design efficient sampling methods for sampling
venues on a map~\cite{Li2012,Li2014,Wang2014}.
Since we can sample venues within an area easily, we are able to
\emph{indirectly} sample Weibo users in an area by {\em relating users to
venues through check-in relationships between them}.
This will be more efficient than directly sampling users in an area.
We leave the detailed design of this sampling method in
Section~\ref{sec:methods} and evaluate it in Section~\ref{sec:experiment}.

More than solving a particular problem in Example~\ref{exam:weibo}, 
we are inspired to study a more generalized problem.
If we consider the venues in Example~\ref{exam:weibo} as another type of
nodes besides user nodes, we can build three graphs, i.e.,
(1) a user graph formed by users and their relationships,
(2) a venue graph formed by venues and their relationships (the edge set
can be empty as in Example~\ref{exam:weibo}),
and (3) a bipartite graph formed by users, venues and their check-in
relationships.
What we learned from Example~\ref{exam:weibo} is that,
when \emph{directly} sampling the user graph is very difficult or extremely
inefficient, we can try to sample the venue graph (which is easier as in
Example~\ref{exam:weibo}), and the bipartite graph acts as a bridge to
connect them.
This approach facilitates us to sample user graph \emph{indirectly} but
efficiently.
Because the \emph{affiliation relationship} between users and venues
plays an important role in these graphs, we refer to the three graphs as
a {\em hybrid social-affiliation network} jointly.
The formal definition of hybrid social-affiliation network will be given
in Section~\ref{sec:problem}, and the detailed design of sampling
methods on hybrid social-affiliation networks will be depicted in
Section~\ref{sec:methods}.

\smallskip\noindent\textbf{Contributions:}
Overall, we have three main contributions:
\begin{itemize}
\item (\emph{Problem Novelty}) We define the idea of hybrid
social-affiliation network and formulate a sampling problem over it.
(Section~\ref{sec:problem}).
\item (\emph{Solution Novelty}) We design three efficient sampling methods
over such a network.
These methods facilitate us to indirectly sample a graph efficiently when
directly sampling it is difficult (Section~\ref{sec:methods}).
\item We conduct extensive experiments to validate the proposed methods
over both synthetic and real-world networks (Section~\ref{sec:experiment}).
\end{itemize}

\section{\textbf{Problem Formulation}} \label{sec:problem}

In this section, we first define the graph characteristics that we want to
measure in this work, and then formally define the hybrid
social-affiliation network along with the sampling problem over it.

\subsection{Graph Characteristics}

We model an OSN by an undirected graph $G(\CU,\CE)$, where $\CU$ and
$\CE$ are the sets of users and relations among users, respectively.
Users in $G$ are labeled. Let $\CL=\{l_1,\ldots,l_W\}$ be a set of user
labels of size $W$.
We map each user $u\in\CU$ to a subset of labels he owned by a set function
called {\em characteristic function} $L\colon\CU\mapsto 2^\CL$.
For example, if $\CL=\{\text{male},\text{female}\}$, then $L(u)$ represents
the gender of $u$.

In many applications, we are interested in measuring the fractions of users
having some labels, e.g., the fraction of male/female customers buying a
product. This can be represented by the label distribution
$\theta=\{\theta_l\}_{l\in\CL}$, where $\theta_l$ is the fraction of users
with label $l$.
That is
\[
\theta_l=\frac{1}{n}\sum_{u\in\CU}\indr{l\in L(u)},\quad l\in\CL,
\]
where $n=|\CU|$ is the size of graph $G$, and $\indr{\cdot}$ is the 
indicator function. When the graph size $n$ is known or can be
estimated~\cite{Katzir2011,Hardiman2013}, we can also obtain the absolute
volume of users having a label $l$ by $n\theta_l$.

With this definition of graph characteristics, the objective of sampling
then becomes how to collect samples (i.e., nodes) from graph $G$ and design
estimators to estimate parameters $\{\theta_l\}_{l\in\CL}$ based on these
samples.

\subsection{Hybrid Social-Affiliation Networks}

Example~\ref{exam:weibo} motivates us to define a {\em hybrid
social-affiliation network} when directly sampling graph $G$ is difficult
or inefficient.
A hybrid social-affiliation network consists of three graphs:
$G(\CU,\CE)$, $G'(\CV,\CE')$, and $G_b(\CU,\CV,\CE_b)$, where
$\CU,\CV$ are sets of nodes and $\CE,\CE',\CE_b$ are sets of edges. 
In detail,
\begin{itemize}
\item $G(\CU,\CE)$ is the \emph{target graph} whose characteristics
$\theta$ are of interest and need to be measured, e.g., the user social
network in Example~\ref{exam:weibo}.
\item $G'(\CV,\CE')$ is an \emph{auxiliary graph} which can be sampled
more easily or efficiently than sampling the target graph, e.g.,
the venue graph (with $\CE'=\emptyset$) in Example~\ref{exam:weibo}.
\item $G_b(\CU,\CV,\CE_b)$ is an \emph{affiliation graph}
\cite[Chapter 8]{Wasserman1994} which is a bipartite graph connecting 
nodes in the target and auxiliary graphs, e.g., the graph formed by
users, venues and their check-in relationships in Example~\ref{exam:weibo}.
\end{itemize}

\begin{figure}
\centering
\begin{tikzpicture}[
und/.style={draw,thick,circle,minimum size=6pt,inner sep=0},
vnd/.style={draw=blue,thick,rectangle,minimum size=6pt,inner sep=0},
att/.style={left,fill=white,minimum size=0,inner sep=0,anchor=west}]
\draw[thick] (0,0) -- (4,0) -- (6,1.2) -- (2,1.2) -- (0,0);
\node[und] (u11) at (1.3,0.2) {};
\node[und,above right = 0.5 and 0.3 of u11] (u12) {};
\node[und,above right = 0.4 and 1.0 of u11] (u13) {};
\node[und,right = 1 of u11] (u14) {};
\node[und,above right = 0.1 and 0.5 of u11] (u15) {};
\draw[thick] (u11)--(u12)--(u13)--(u14)--(u11)--(u15)--(u12) (u15)--(u14);
\node[und] (u21) at (3.2,0.35) {};
\node[und,above right = 0.5 and 0.5 of u21] (u22) {};
\node[und,above right = 0.4 and 1.2 of u21] (u23) {};
\node[und,right = 0.7 of u21] (u24) {};
\node[und,above right = 0.1 and 0.5 of u21] (u25) {};
\draw[thick] (u21)--(u22)--(u23)--(u25)--(u24)--(u21) (u22)--(u25)
(u23)--(u24);
\node[vnd](v1) at (1.5,2.1) {};
\node[vnd,above right = 0.4 and 0.6 of v1] (v2) {};
\node[vnd,above right = 0.3 and 2.5 of v1] (v3) {};
\node[vnd,right = 1.5 of v1] (v4) {};
\draw[thick,blue] (v1) -- (v2) -- (v3) -- (v4) -- (v2);
\draw[blue,thick] (0,1.7) -- (4,1.7) -- (6,3.0) -- (2,3.0) -- (0,1.7);
\draw[dashed,red,thick] (u11)--(v1) (v2)--(u12) (u14)--(v1)
(u15)--(v2)--(u13);
\draw[dashed,red,thick] (v2)--(u25) (v3)--(u22)--(v4)--(u21)
(u24)--(v3)--(u23);
\node[att] at (.4,0.13) {$G$};
\node[att] at (.4,1) {$\color{red} G_b$};
\node[att] at (.4,1.85) {$\color{blue} G'$};
\node[att] at (5,1.0) {target graph};
\node[att] at (5,2.8) {\textcolor{blue}{auxiliary graph}};
\node[att] at (5,1.9) {\textcolor{red}{affiliation graph}};
\end{tikzpicture}

\caption{An illustration of a hybrid social-affiliation network.
The target graph together with auxiliary and affiliation graphs form a
better connected graph than target graph itself, which improves sampling
efficiency.}
\label{fig:exam}
\end{figure}
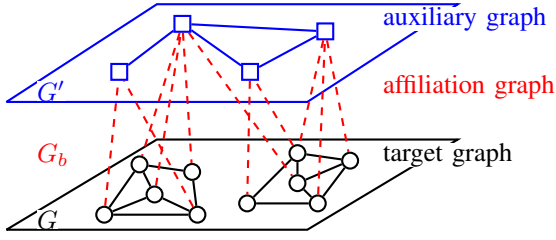

An example of such a hybrid social-affiliation network is given in
Fig.~\ref{fig:exam}.
In addition to Example~\ref{exam:weibo}, many other measuring problems can
be formed as a hybrid social-affiliation network sampling problem.
As another case, let us consider the following example.

\begin{example}\label{exam:mtime}
Mtime.com~\cite{Mtime} is an online movie database in China.
Users in Mtime can follow each other to form a social network. 
Moreover, a user can follow movie actors if he is a fan of the actor.
The movie actors can also form connections if they cooperated in a same
movie.
\end{example}

In Example~\ref{exam:mtime}, if we want to measure the characteristics of
the graph formed by Mtime users, and directly sampling users is inefficient
(because of the community structure formed by user interests difference,
geographic constraints etc., which make the user graph not well connected)
we can build a hybrid social-affiliation network as follows:
\begin{itemize}
\item\emph{The target graph} is formed by Mtime users and their following
relationships.
\item\emph{The auxiliary graph} is formed by actors and their cooperation
relationships.
\item\emph{The affiliation graph} is formed by Mtime users and actors and
the fan relationships between them.
\end{itemize}
Other than the ordinary people, movie actors especially pop stars are more
easily to form connections since they have more chances to join same events
such as Oscar and Cannes.
That is, the auxiliary graph is more likely to be well connected than the
target graph.
We can leverage this feature to design efficient sampling methods to
measure target graph characteristics.

\section{\textbf{Sampling Design on Hybrid Social-Affiliation Networks}}
\label{sec:methods}

In this section, we design three sampling methods for measuring target
graph characteristics on a hybrid social-affiliation network.
The notations that will be used in this section are summarized in
Table~\ref{tab:notations}.

\begin{table}
\centering
\caption{Notations\label{tab:notations}}
\begin{tabular}{c|l}
\hline\hline
$G,G',G_b$  &  target/auxiliary/affiliation graph.\\
$\CU,\CV$  &  sets of nodes. \\
$n,n'$ & size of target/auxiliary graph, i.e., $n\!=\!|\CU|,n'\!=\!|\CV|$.\\
$\CE,\CE',\CE_b$ & sets of edges. \\
$\CS,\CS'$ & sets of node samples in target and auxiliary graphs.\\
$B,B'$ & sampling budgets, i.e., $B=|\CS|,B'=|\CS'|$. \\
$\CV_u,\CU_v$ & neighbors of node $u$ or $v$ in the graph.\\
$d_u,d_v$ & degree of node $u$ (or $v$) in target (or auxiliary) graph.\\
$d_u^{(b)},d_v^{(b)}$ & degree of node $u$ (or $v$) in affiliation graph.\\
\hline
\end{tabular}
\end{table}

\subsection{Indirectly Sampling Target Graph by Vertex Sampling on 
Auxiliary Graph (\VS{A})}

When random vertex sampling is more easily to be conducted on auxiliary
graph than on target graph such as the case in Example~\ref{exam:weibo}, we
propose a sampling method \VS{A} to randomly sample vertices in auxiliary
graph so as to indirectly sample target graph.
The basic idea of \VS{A} is illustrated in Fig.~\ref{fig:vsa}.

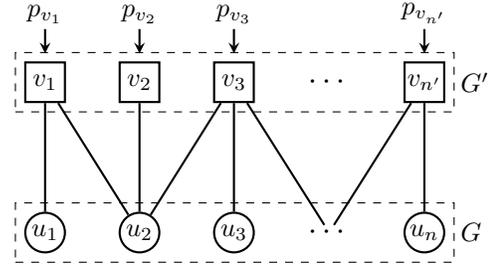
\begin{figure}[htp]
\centering
\begin{tikzpicture}[
und/.style={draw,thick,circle,minimum size=15pt,inner sep=0},
vnd/.style={draw,thick,rectangle,minimum size=15pt,inner sep=0},
att/.style={left,minimum size=0,inner sep=0},
arr/.style={thick,->,>=stealth}]
\node(u1) [und] at (0,0) {$u_1$};
\node(u2) [und,right=.7] at (u1.east) {$u_2$};
\node(u3) [und,right=.7] at (u2.east) {$u_3$};
\node(p1) [att,right=.7] at (u3.east) {\large $\cdots$};
\node(un) [und,right=.7] at (p1.east) {$u_n$};
\node[draw,dashed,fit={(u1) (un)}] {};
\node[att,right=2mm] at (un.east) {$G$};

\node(v1) [vnd] at (0,2) {$v_1$};
\node(v2) [vnd,right=.7] at (v1.east) {$v_2$};
\node(v3) [vnd,right=.7] at (v2.east) {$v_3$};
\node(p2) [att,right=.7] at (v3.east) {\large $\cdots$};
\node(vn) [vnd,right=.7] at (p2.east) {$v_{n'}$};
\node[draw,dashed,fit={(v1) (vn)}] {};
\node[att,right=2mm] at (vn.east) {$G'$};

\draw[thick] (u2)--(v1)--(u1) (v2)--(u2)--(v3)--(p1)--(vn)--(un) (v3)--(u3);

\foreach \v/\a in {v1/1,v2/2,v3/3,vn/n'}{
  \node[att,above=0.5 of \v] (tmp) {$p_{v_{\a}}$};
  \draw[arr] ([yshift=-2pt]tmp.south)--([yshift=3pt]\v.north);
}
\end{tikzpicture}

\caption{\VS{A}. Edges in target and auxiliary graphs are omitted.}
\label{fig:vsa}
\end{figure}

Suppose a node $v\in\CV$ is sampled with probability $p_v$ in $G'$.
For example, when graph $G'$ supports the uniform vertex sampling, then
$p_v=1/n', \forall v\in\CV$, where $n'=|\CV|$ is the size of graph $G'$.

\smallskip\noindent\textbf{Sampling Design.}
The sampling design of \VS{A} consists of the following two steps:
\begin{description}[\IEEEsetlabelwidth{\emph{Step (ii)}}]
\item[\emph{Step (i)}] Sampling a collection of $B'$ nodes with replacement
  in auxiliary graph $G'$, and denote these samples by
  $\CS'=\{y_1,\ldots,y_{B'}\}$.
\item[\emph{Step (ii)}] For each $v\in\CS'$, let $\CU_v\subseteq\CU$ be
  the subset of nodes that are connected to $v$ in $G_b$, and nodes in
  $\CU_v$ are all included into $\CS$, i.e., $\CS=\CS\cup\CU_v$.
\end{description}

Having collected samples $\CS$ in the target graph, \VS{A} uses $S$ to
estimate target graph characteristics.

\smallskip\noindent\textbf{Estimators.}
When $n=|\CU|$ is known in advance, we can use the following estimator
to estimate $\theta_l$, 
\begin{equation}
\hat\theta_l^{\VS{A}} = \frac{1}{nB'}\sum_{i=1}^{B'}\frac{1}{p_{y_i}}
\sum_{u\in\CU_{y_i}}\frac{\indr{l\in L(u)}}{d_u^{(b)}},
\label{eq:es1}
\end{equation}
where $d_u^{(b)}$ is the degree of node $u$ in affiliation graph $G_b$.
When $n$ is unknown, we can estimate $n$ by
\begin{equation}
\hat{n} = \frac{1}{B'}\sum_{i=1}^{B'}\frac{1}{p_{y_i}}
\sum_{u\in\CU_{y_i}}\frac{1}{d_u^{(b)}},
\label{eq:nes}
\end{equation}
and another estimator for $\theta_l$ when $n$ is unknown is
\begin{equation}
\check\theta_l^{\VS{A}}=\frac{1}{\hat{n}B'}\sum_{i=1}^{B'}\frac{1}{p_{y_i}}
\sum_{u\in\CU_{y_i}}\frac{\indr{l\in L(u)}}{d_u^{(b)}}.
\label{eq:es2}
\end{equation}

The following theorem guarantees the \emph{unbiasedness} of these estimators.

\begin{theorem}
Estimators~\eqref{eq:es1} and \eqref{eq:nes} are unbiased estimators of
$\theta_l$ and $n$, respectively.
Estimator~\eqref{eq:es2} is an asymptotically unbiased estimator of
$\theta_l$.
\end{theorem}
\begin{proof}
We show that
\begin{align*}
\E{\hat\theta_l^{\VS{A}}} &= \frac{1}{nB'}\sum_{i=1}^{B'}\E{\frac{1}{p_{y_i}}
\sum_{u\in\CU_{y_i}}\frac{\indr{l\in L(u)}}{d_u^{(b)}}} \\
&= \frac{1}{n}\sum_{v\in\CV}p_v\frac{1}{p_v}\sum_{u\in\CU_v}
\frac{\indr{l\in L(u)}}{d_u^{(b)}} \\
&= \frac{1}{n}\sum_{u\in\CU}\indr{l\in L(u)} \\
&= \theta_l.
\end{align*}
The second equality holds because that $y_i,i=1,\ldots,B'$ are i.i.d random
variables.
The third equality holds because that each item in the inner summation is
added $d_u^{(b)}$ times for each $u\in\CU$. 
Hence, $\hat\theta_l^{\VS{A}}$ is unbiased.

In a similar manner, we can prove that estimator~\eqref{eq:nes} is an
unbiased estimator of $n$, which we omit here.

To prove that estimator~\eqref{eq:es2} is asymptotically unbiased, we use
the ratio form of the law of large numbers in \cite[Theorem~17.2.1 on
P.~428]{Meyn2009}. Hence
\[
\lim_{B'\rightarrow\infty}\check\theta_l^{\VS{A}} 
=\frac{\E{n\hat\theta_l^{\VS{A}}}}{\E{\hat{n}}}
=\theta_l.
\]
\end{proof}

It is important to note that \VS{A} can only sample nodes in $\CU$
satisfying $d_u^{(b)}>0$ in the target graph.
Because a node in $G$ having no connection to nodes in $G'$ can not be
indirectly sampled according to the design of \VS{A}.
In Example~\ref{exam:weibo}, since we are only interested in users who have
check-ins in Weibo, therefore Example~\ref{exam:weibo} satisfies this
condition.

\subsection{Random Walking on Target Graph with Vertex Sampling on
Auxiliary Graph (\RW{T}\VS{A})}

In some situations, $d_u^{(b)}=0$ for some $u\in\CU$.
For example, some user nodes in Example~\ref{exam:mtime} may not follow any
movie actors at all, and these users cannot be sampled by \VS{A}.
To overcome this problem, we design another sampling method \RW{T}\VS{A}
which combines random walk sampling on target graph and vertex sampling on
auxiliary graph.

The basic idea of \RW{T}\VS{A} is that, we run a simple random walk on the
target graph, and at each step the random walk jumps with a probability
related to the node that it currently resides.
The node to jump to is randomly chosen from neighbors of $v$ in the
affiliation graph and $v$ is randomly sampled in the auxiliary graph.
We can show that this approach is equivalent to the standard
RWwJ~\cite{Avrachenkov2010,Ribeiro2012b} on $G$, and this idea is
illustrated in Fig.~\ref{fig:rwtvsa}.
An additional advantage of running random walk on target graph is that
a random walk can better characterize highly connected nodes than uniform
sampling as random walks are biased to sample high degree nodes in $G$.

\begin{figure}[htp]
\centering
\begin{tikzpicture}[
und/.style={draw,thick,circle,minimum size=15pt,inner sep=0},
vnd/.style={draw,thin,rectangle,minimum size=15pt,inner sep=0},
att/.style={left,minimum size=0,inner sep=0},
arr/.style={thick,->,>=stealth},
jumper/.style={draw,thick,fill=gray,fill opacity=.7,text opacity=1,regular
polygon,regular polygon sides=6,minimum size=16pt,inner sep=0}]
\node(u1) [und] at (0,0) {$u_1$};
\node(u2) [und,right=.7] at (u1.east) {$u_2$};
\node(u3) [und,right=.7] at (u2.east) {$u_3$};
\node(p1) [att,right=.7] at (u3.east) {\large $\cdots$};
\node(un) [und,right=.7] at (p1.east) {$u_n$};
\node(j1) [jumper,above=1 of u3] {$j$};
\node[draw,dashed,fit={(u1) (un)}] {};
\node[att,right=2mm] at (un.east) {$G$};
\node[att,right=2mm] at (j1.east) {Jumper node};

\foreach \u in {1,2,3,n}{
  \draw[thick] (j1) to node[att] {\small $\omega_{u_\u}$} (u\u);
}

\node(v1) [vnd,above=1.9 of u1] {$v_1$};
\node(v2) [vnd,right=.7] at (v1.east) {$v_2$};
\node(v3) [vnd,right=.7] at (v2.east) {$v_3$};
\node(p2) [att,right=.7] at (v3.east) {\large $\cdots$};
\node(vn) [vnd,right=.7] at (p2.east) {$v_{n'}$};
\node[draw,dashed,fit={(v1) (vn)}] {};
\node[att,right=2mm] at (vn.east) {$G'$};

\foreach \v/\a in {v1/1,v2/2,v3/3,vn/n'}{
  \node[att,above=0.5 of \v] (tmp) {$p_{v_{\a}}$};
  \draw[arr] ([yshift=-2pt]tmp.south)--([yshift=3pt]\v.north);
}

\end{tikzpicture}

\caption{\RW{T}\VS{A}. Edges in target and auxiliary graphs are omitted.}
\label{fig:rwtvsa}
\end{figure}

As in \VS{A}, we assume a node $v\in\CV$ can be sampled with probability
$p_v$.
In \RW{T}\VS{A}, we virtually connect each node $u\in\CU$ to a \emph{jumper}
node $j$ by edge $(u,j)$, and each edge $(u,j)$ is assigned with a weight
$\omega_u=\alpha q_u$, where $\alpha$ is a constant controlling the
probability of jumping at each step of random walk, and $q_u$ is determined
as follows in order to protect the reversible property of Markov chain.
\begin{equation}
q_u=\sum_{v\in\CV_u}\frac{p_v}{d_v^{(b)}},\ u\in\CU,
\label{eq:qu}
\end{equation}
where $\CV_u$ is the subset of nodes in $\CV$ that are connected to $u$,
and $d_v^{(b)}$ is the degree of node $v$ in affiliation graph $G_b$.
The random walk jumps from a node $u$ to jumper $j$ with probability
\[
p_{uj}=\frac{\omega_u}{d_u+\omega_u},
\]
and moves from the jumper $j$ to $u$ with probability
\[
p_{ju}=\frac{\omega_u}{\sum_{u'}\omega_{u'}}=q_u.
\]
Note that, if $d_u^{(b)}=0$, then $p_{uj}=p_{ju}=0$ according to
Eq.~\eqref{eq:qu}.
So the random walk does not jump from $u$ or to $u$ if $d_u^{b}=0$.

\RW{T}\VS{A} exhibits similar properties as RWwJ.
That is, when $\alpha=0$, \RW{T}\VS{A} becomes a simple random walk on the
target graph.
When $\alpha=\infty$, \RW{T}\VS{A} is equivalent to \VS{A} and it is also
equivalent to random vertex sampling on the target graph with probability
distribution $\{q_u\}_{u\in\CU}$.

When \RW{T}\VS{A} reaches the steady state, each node $u$ is sampled
with probability
\begin{equation}
\pi_u=\frac{d_u+\omega_u}{2|\CE|+\alpha}, \ u\in\CU.
\label{eq:rwtvsa_stable}
\end{equation}

\smallskip\noindent\textbf{Sampling Design.}
Suppose the random walk starts at node $x_1\in\CU$, and at step $i$ the
random walker is at node $x_i$.
We calculate the probability $q_{x_i}$ according to Eq.~\eqref{eq:qu} and
$\omega_{x_i}=\alpha q_{x_i}$.
At step $i$, the walker jumps with probability
$\omega_{x_i}/(d_{x_i}+\omega_{x_i})$; otherwise, the walker moves to a
neighbor $u$ of $x_i$ chosen uniformly at random and set $x_{i+1}=u$.
The jump is conducted as follows:
\begin{description}[\IEEEsetlabelwidth{\emph{Step (ii)}}]
\item[\emph{Step (i)}] We sample a node $v\in\CV$ in the auxiliary graph
with probability $p_v$.
\item[\emph{Step (ii)}] We sample a neighbor $u$ of $v$ uniformly at random
in the affiliation graph, and let $x_{i+1}=u$.
\end{description}

\smallskip\noindent\textbf{Estimator.}
According to the stationary distribution Eq.~\eqref{eq:rwtvsa_stable} of
\RW{T}\VS{A}, we can use the sample path $\{x_i\}_{i=1}^B$ by the random
walk to design a Hansen-Hurwitz estimator of $\{\theta_l\}_{l\in\CL}$ as
follows,
\begin{equation}
\hat\theta_l^{\RW{T}\VS{A}} = \frac{1}{Z}\sum_{i=1}^B\frac{\indr{l\in
L(x_i)}}{d_{x_i}+\omega_{x_i}},
\label{eq:es_rwtvsa}
\end{equation}
where $Z=\sum_{i=1}^B 1/(d_{x_i}+\omega_{x_i})$.

\begin{theorem}\label{th:rwtvsa}
Estimator~\eqref{eq:es_rwtvsa} is an asymptotically unbiased estimator of
$\theta_l$.
\end{theorem}
\begin{proof}
Let $D\triangleq \sum_{i=1}^B\indr{l\in
L(x_i)}/(d_{x_i}+\omega_{x_i})$.
Then
\begin{align*}
\E{D} 
&= B\E{\frac{\indr{l\in L(x_i)}}{d_{x_i}+\omega_{x_i}}} \\
&= B\sum_{u\in\CU}\pi_u\frac{\indr{l\in L(u)}}{d_u+\omega_u} \\
&= \frac{B}{2|\CE|+\alpha}\sum_{u\in\CU}\indr{l\in L(u)} \\
&= \frac{Bn}{2|\CE|+\alpha}\theta_l.
\end{align*}
Similarly, we can show that
\begin{align*}
\E{Z} 
&= \E{\sum_{i=1}^B \frac{1}{d_{x_i}+\omega_{x_i}}} \\
&= B\E{\frac{1}{d_{x_i}+\omega_{x_i}}} \\
&= B\sum_{u\in\CU}\pi_u\frac{1}{d_u+\omega_u} \\
&= \frac{Bn}{2|\CE|+\alpha}.
\end{align*}
Now, we invoke Theorem~17.2.1 in \cite[P.~428]{Meyn2009}, which is the
ratio form of the law of large numbers, and indicate that
\[
\lim_{B\rightarrow\infty}\hat\theta_l^{\RW{T}\VS{A}}
=\frac{\E{D}}{\E{Z}}
=\theta_l.
\]
\end{proof}

Note that \RW{T}\VS{A} requires that we can conduct vertex sampling on 
auxiliary graph $G'$.
In fact, we can replace vertex sampling by another simple random walk on
auxiliary graph $G'$.
However, this simple random walk may be easily trapped when $G'$ is not
well connected.
In the follows, we design a new method to address this problem.

\subsection{Random Walking on Target Graph with Random Walking on
Auxiliary Graph (\RW{T}\RW{A})}

When both the target and auxiliary graphs do not support random vertex
sampling, neither \VS{A} nor \RW{T}\VS{A} can be applied under this
situation. 
Therefore, we design the \RW{T}\RW{A} method in this subsection, namely,
random walking on the target graph with random walking on the auxiliary
graph.
\RW{T}\RW{A} consists of two parallel random walks on $G$ and $G'$
respectively.
The two parallel random walks cooperate with each other, and can be
considered as two random walks with jumps, as illustrated in
Fig.~\ref{fig:rwtrwa}.
Nodes in $G$ and $G'$ are virtually connected to two jumper nodes $j_1$ and
$j_2$, respectively.

\begin{figure}[htp]
\centering
\begin{tikzpicture}[
und/.style={draw,thick,circle,minimum size=15pt,inner sep=0},
vnd/.style={draw,thick,rectangle,minimum size=15pt,inner sep=0},
jumper/.style={draw,thick,fill=gray,fill opacity=.7,text opacity=1,regular polygon,regular polygon sides=6,minimum size=16pt,inner sep=0},
att/.style={left,minimum size=0,inner sep=0}]
\node(u1) [und] at (0,0) {$u_1$};
\node(u2) [und,right=.7] at (u1.east) {$u_2$};
\node(u3) [und,right=.7] at (u2.east) {$u_3$};
\node(p1) [att,right=.7] at (u3.east) {\large $\cdots$};
\node(un) [und,right=.7] at (p1.east) {$u_n$};
\node(j1) [jumper,above left=1.08 and 0.1 of u3] {$j_1$};
\node[draw,dashed,fit={(u1) (un)}] {};
\node[att,right=2mm] at (un.east) {$G$};
\foreach \v in {1,2,3,n}{
	\draw[thick] (j1) to node[att] {$\omega_{u_{\v}}$} (u\v);
}
\node(v1) [vnd] at (0,3) {$v_1$};
\node(v2) [vnd,right=.7] at (v1.east) {$v_2$};
\node(v3) [vnd,right=.7] at (v2.east) {$v_3$};
\node(p2) [att,right=.7] at (v3.east) {\large $\cdots$};
\node(vnp) [vnd,right=.7] at (p2.east) {$v_{n'}$};
\node(j2) [jumper,below right=1.08 and 0.1 of v3] {$j_2$};
\node[draw,dashed,fit={(v1) (vnp)}] {};
\node[att,right=2mm] at (vnp.east) {$G'$};
\node[att,right=2mm] at (j2.east) {Jumper nodes};
\foreach \v/\a in {1/1,2/2,3/3,np/n'} {
	\draw[thick] (j2) to node[att] {$w_{v_{\a}}$} (v\v);
}
\end{tikzpicture}

\caption{\RW{T}\RW{A}. Edges in target and auxiliary graphs are omitted.}
\label{fig:rwtrwa}
\end{figure}

The basic idea behind \RW{T}\RW{A} is as follows.
Suppose the two random walks are $RW$ on $G$ and $RW'$ on $G'$
respectively, and they are at $x_i$ and $y_i$ at step $i$. 
If one random walk needs to jump at step $i$, say $RW$, then the node to
jump to is randomly chosen from $y_i$'s neighbors in the affiliation graph
and assigned to $x_{i+1}$.
Similar jumping procedure also applies to $RW'$.
Therefore, they are equivalent to two RWwJs.

The main problem we need to solve is how to determine edge weights
$\omega_\CU\triangleq\{\omega_u\}_{u\in\CU}$ and
$w_\CV\triangleq\{w_v\}_{v\in\CV}$,
which control the probability of jumping on $G$ and $G'$ respectively.
Obviously, the stationary distributions of the two random walks are related
to these weights.
Let $\pi_u$ and $\pi_v$ be the stationary distributions of sampling nodes
in $G$ and $G'$ respectively.
They are determined by
\begin{align}
\pi_u &= \frac{d_u+\omega_u}{2|\CE|+\alpha}, \ u\in\CU,
\label{eq:piu} \\
\pi_v &= \frac{d_v+w_v}{2|\CE'|+\beta}, \ v\in\CV,
\label{eq:piv}
\end{align}
where $\alpha$ and $\beta$ are two constants.
These weights should be assigned properly so as to keep the reversibility
of Markov chains.
Therefore, stationary distributions control $\omega_u$ and $w_v$ in turn.
\begin{align}
\omega_u &= \alpha\sum_{v\in\CV_u}\frac{\pi_v}{d_v^{(b)}}, \ u\in\CU,
\label{eq:wu} \\
w_v &= \beta\sum_{u\in\CU_v}\frac{\pi_u}{d_u^{(b)}}, \ v\in\CV.
\label{eq:wv}
\end{align}
Or we can arrange Eqs.~\eqref{eq:piu}-\eqref{eq:wv} in matrix formulas,
\begin{align*}
\pi_\CU &= \frac{d_\CU+\omega_\CU}{2|\CE|+\alpha}, 
&\pi_\CV &= \frac{d_\CV+w_\CV}{2|\CE'|+\beta}, \\
\omega_\CU &=\alpha AD_\CV^{-1}\pi_\CV, 
&w_\CV &= \beta A^TD_\CU^{-1}\pi_\CU,
\end{align*}
where $A=[a_{uv}]_{n\times n'}$ is the adjacency matrix of $G_b$, 
$\pi_\CU=[\pi_u]_{u\in\CU}^T, \pi_\CV=[\pi_v]_{v\in\CV}^T$,
$\omega_\CU=[\omega_u]_{u\in\CU}^T, w_\CV=[w_v]_{v\in\CV}^T$,
$d_\CU=[d_u]_{u\in\CU}^T, d_\CV=[d_v]_{v\in\CV}^T$ are six vectors,
and $D_\CU=diag(d_{u_1}^{(b)},\ldots,d_{u_n}^{(b)}), 
D_\CV=diag(d_{v_1}^{(b)},\ldots,d_{v_{n'}}^{(b)})$.

Above equations can uniquely determine $\omega_\CU$ and $w_\CV$, i.e.,
\begin{align*}
\omega_\CU &\!=\! c'(I\!-\!cc'AD_\CV^{-1}A^TD_\CU^{-1})^{-1}\!
AD_\CV^{-1}(d_\CV\!+\!cA^TD_\CU^{-1}d_\CU), \\
w_\CV &\!=\! c(I\!-\!cc'A^TD_\CU^{-1}AD_\CV^{-1})^{-1}\!
A^TD_\CU^{-1}(d_\CU\!+\!c'AD_\CV^{-1}d_\CV),
\end{align*}
where $c=\beta/(2|\CE|+\alpha)$ and $c'=\alpha/(2|\CE'|+\beta)$ are two
constants.

Above results illustrate that, given $\alpha$ and $\beta$, $\omega_\CU$
and $w_\CV$ are uniquely determined.
However, they need complete knowledges of $G$, $G'$ and $G_b$
to determine their precise values.
This feature is not suitable for us to design an algorithm that only uses
local information of these graphs.
In what follows, we address this problem and design \RW{T}\RW{A} in a way
that only requires local knowledges of these graphs.

Suppose we firstly fix $\omega_u=\alpha q_u$, e.g., specify $q_u$ to
follow a uniform distribution over $\CU$.
Using above equations, we can determine $\pi_u$, $w_v$ and $\pi_v$ in order.
Then, using Eq.~\eqref{eq:wu}, we can calculate a new $\omega_u'$ which may
not equal to $\omega_u$.
Let $\omega_u'=\alpha q_u'$.
Since $q_u\neq q_u'$, this will cause the Markov chain on $G$ to be
non-reversible.
To address this problem, we apply Metropolis-Hastings sampler~\cite[Chapter
7]{Robert2004} by considering $\{q_u\}_{u\in\CU}$ as the \emph{desired
distribution} and $\{q_u'\}_{u\in\CU}$ as the \emph{proposal distribution}.
Therefore, we can use Metropolis-Hastings sampler to build a Markov chain
(refer as MH chain) that can generate samples with the desired distribution
$\{q_u\}_{u\in\CU}$, and each time when the random walk on $G$ requires to
jump, it jumps to a sample of MH chain, thereby preserving the
reversibility of Markov chain on $G$.

\begin{algorithm}
\caption{Metropolis-Hastings Sampler.\label{alg:MH-sampler}}
\KwIn{A desired distribution $\{q_u\}_{u\in\CU}$, and a proposed
distribution $\{q_u'\}_{u\in\CU}$.}
\KwOut{A Markov chain with the desired stationary distribution.}
Let $u_t$ be the sample at time $t$\;
\While{not stop,}{
	Draw $u$ from $\{q_u'\}$\;
	Calculate the acceptance ratio
	$r_t=\min\{1,\frac{q_uq_{u_t}'}{q_{u_t}q_u'}\}$\;
	Set $u_{t+1}=u$ with probability $r_t$\;
	Set $u_{t+1}=u_t$ with probability $1-r_t$\;
}
\end{algorithm}

\smallskip\noindent\textbf{Sampling Design.}
We specify a desired sampling distribution $\{q_u\}_{u\in\CU}$ over
$\CU$, e.g., a uniform distribution.
The complete sampling design of \RW{T}\RW{A} comprises three Markov chains
as shown in Fig.~\ref{fig:mcs}.

\smallskip\noindent\,\textbullet\emph{Random Walk on Auxiliary Graph:}
Suppose the random walker resides at node $y_i\in\CV$ at step $i$.
Then we can easily calculate $w_{y_i}$ according to Eq.~\eqref{eq:wv}.
At step $i+1$, the walker execute one of the following two steps.
\begin{description}[\IEEEsetlabelwidth{\emph{Jumping}}]
\item[\emph{Jumping}] With probability
$w_{y_i}/(d_{y_i}+w_{y_i})$, the walker jumps to a random
neighbor $v\in\CV$ of node $x_i$ in $G_b$, and set $y_{i+1}=v$; 
\item[\emph{Walking}] Otherwise, the walker moves to a random neighbor
$v\in\CV$ of $y_i$ in $G'$, and set $y_{i+1}=v$.
\end{description}

\smallskip\noindent\,\textbullet\emph{Metropolis-Hastings (MH) Chain:}
Suppose the MH chain resides at node $x'_i$ at step $i$.
At step $i+1$, we randomly choose a neighbor $u\in\CU$ of $y_i$ in $G_b$.
This is equivalent to sample a node $u\in\CU$ with probability $q_u$.
\begin{description}[\IEEEsetlabelwidth{\emph{Accept}}]
\item[\emph{Accept}] With probability $r_i$, we accept $u$ and set
$x'_{i+1}=u$, where $r_i=\min\{1,(q_uq_{x'_i}')/(q_{x'_i}q_u')\}$ (note
that $q_{x'_i}'$ has been calculated at step $i$).
\item[\emph{Reject}] Otherwise, we reject $u$ and set $x'_{i=1}=x'_i$.
\end{description}

\smallskip\noindent\,\textbullet\emph{Random Walk on Target Graph:}
Suppose the random walker resides at node $x_i\in\CU$ at step $i$.
Then we can easily calculate $\omega_{x_i}$ according to
Eq.~\eqref{eq:wu}.
At step $i+1$, the walker execute one of the following two steps.
\begin{description}[\IEEEsetlabelwidth{\emph{Jumping}}]
\item[\emph{Jumping}] With probability
$\omega_{x_i}/(d_{x_i}+\omega_{x_i})$, the walker jumps to $x'_{i+1}$, and
set $x_{i+1}=x'_{i+1}$; 
\item[\emph{Walking}] Otherwise, the walker moves to a random neighbor
$u\in\CU$ of $x_i$ in $G$, and set $x_{i+1}=u$.
\end{description}

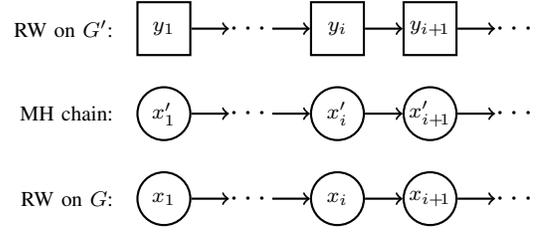
\begin{figure}
\centering
\footnotesize
\begin{tikzpicture}[
und/.style={draw,thick,circle,minimum size=20pt,inner sep=0},
vnd/.style={draw,thick,rectangle,minimum size=20pt,inner sep=0},
att/.style={left,minimum size=0,inner sep=0}]
\node(u1) [und] at (0,0) {$x_1$};
\node(p1) [att,right=.5 of u1] {\large$\cdots$};
\node(ui) [und,right=.5 of p1] {$x_i$};
\node(ui1) [und,right=.5 of ui] {$x_{i+\!1}$}; 
\node(p2) [att,right=.5 of ui1] {\large$\cdots$};
\draw[thick,->] (u1)--(p1);
\draw[thick,->] (p1)--(ui);
\draw[thick,->] (ui)--(ui1);
\draw[thick,->] (ui1)--(p2);
\node[att,left=.3 of u1] {RW on $G$:}; 

\node(up1) [und,above=.4 of u1] {$x'_1$};
\node(pp1) [att,right=.5 of up1] {\large$\cdots$}; 
\node(upi) [und,right=.5 of pp1] {$x'_i$}; 
\node(upi1) [und,right=.5 of upi] {$x'_{i+\!1}$}; 
\node(pp2) [att,right=.5 of upi1] {\large$\cdots$}; 
\draw[thick,->] (up1)--(pp1);
\draw[thick,->] (pp1)--(upi);
\draw[thick,->] (upi)--(upi1);
\draw[thick,->] (upi1)--(pp2);
\node[att,left=.3 of up1] {MH chain:}; 

\node(v1) [vnd,above=.4 of up1] {$y_1$};
\node(pv1) [att,right=.5 of v1] {\large$\cdots$}; 
\node(vi) [vnd,right=.5 of pv1] {$y_i$}; 
\node(vi1) [vnd,right=.5 of vi] {$y_{i+\!1}$}; 
\node(pv2) [att,right=.5 of vi1] {\large$\cdots$}; 
\draw[thick,->] (v1)--(pv1);
\draw[thick,->] (pv1)--(vi);
\draw[thick,->] (vi)--(vi1);
\draw[thick,->] (vi1)--(pv2);
\node[att,left=.3 of v1] {RW on $G'$:}; 
\end{tikzpicture}

\caption{Markov chains in the \RW{T}\RW{A}.} 
\label{fig:mcs}
\end{figure}

\smallskip\noindent\textbf{Estimator.}
We use the sample path $\{x_i\}_{i=1}^B$ by the random walk on $G$ to
design an estimator to estimate $\{\theta_l\}_{l\in\CL}$ as follows,
\begin{equation}
\hat\theta_l^{\RW{T}\RW{A}} = \frac{1}{Z}\sum_{i=1}^B\frac{\indr{l\in
L(x_i)}}{d_{x_i}+\omega_{x_i}},
\label{eq:rwtrwa}
\end{equation}
where $Z=\sum_{i=1}^B 1/(d_{x_i}+\omega_{x_i})$.

\begin{theorem}
Estimator~\eqref{eq:rwtrwa} is an asymptotically unbiased estimator for
$\theta_l$.
\end{theorem}
\begin{proof}
First we note that the random walk on target graph $G$ has the same
stationary distribution as Eq.~\eqref{eq:rwtvsa_stable}.
So the remaining proof is similar to the proof of Theorem~\ref{th:rwtvsa}.
\end{proof}

\section{\textbf{Experiments}} \label{sec:experiment}

In this section, we conduct experiments on both synthetic and real-world
datasets to evaluate the effectiveness of proposed methods in previous
section.
We will use degree distribution as the graph characteristic to be measured
in these experiments.
That is, $\theta_l,l\geq 0$ denotes the fraction of nodes with degree $l$
in the target graph $G$.

\begin{figure*}[htp]
\centering
\subfloat[$\check\theta^{\VS{A}}_2$ and $\check\theta^{\VS{A}}_{12}$]{
	\includegraphics[width=.33\linewidth]{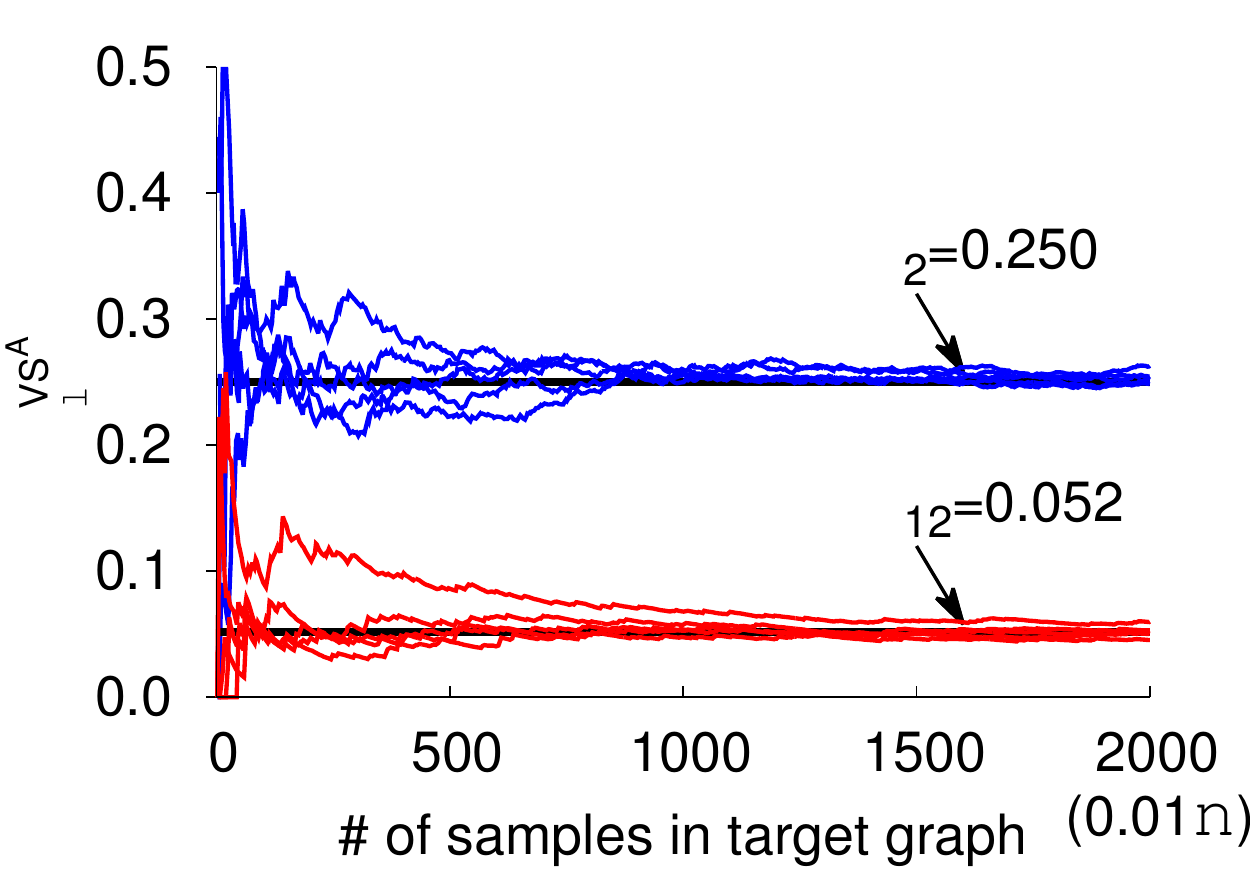}}
\subfloat[$\hat\theta^{\RW{T}\VS{A}}_2$ and
$\hat\theta^{\RW{T}\VS{A}}_{12}$ ($\alpha=10$)]{
	\includegraphics[width=.33\linewidth]{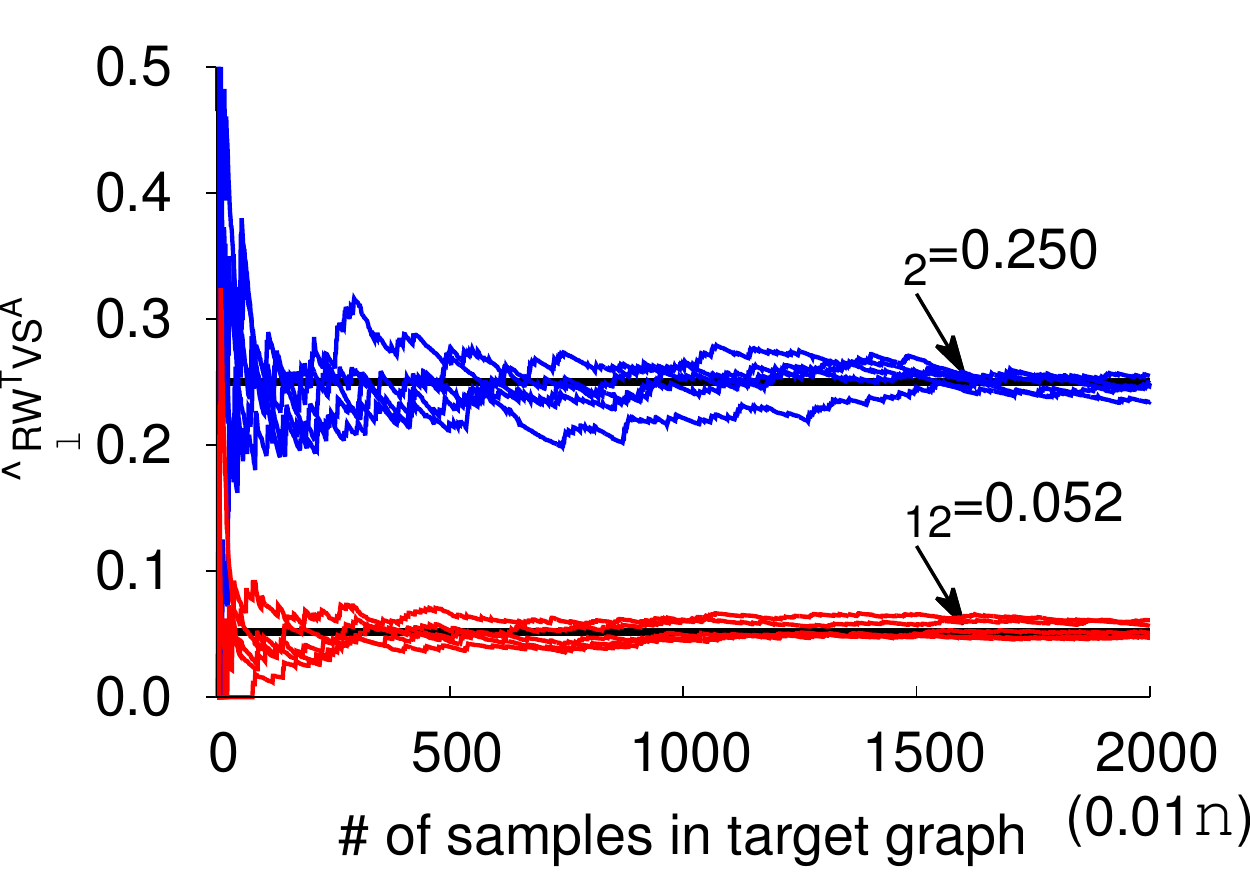}}
\subfloat[$\hat\theta^{\RW{T}\RW{A}}_2$ and
$\hat\theta^{\RW{T}\RW{A}}_{12}$ ($\alpha=\beta=10$)]{
	\includegraphics[width=.33\linewidth]{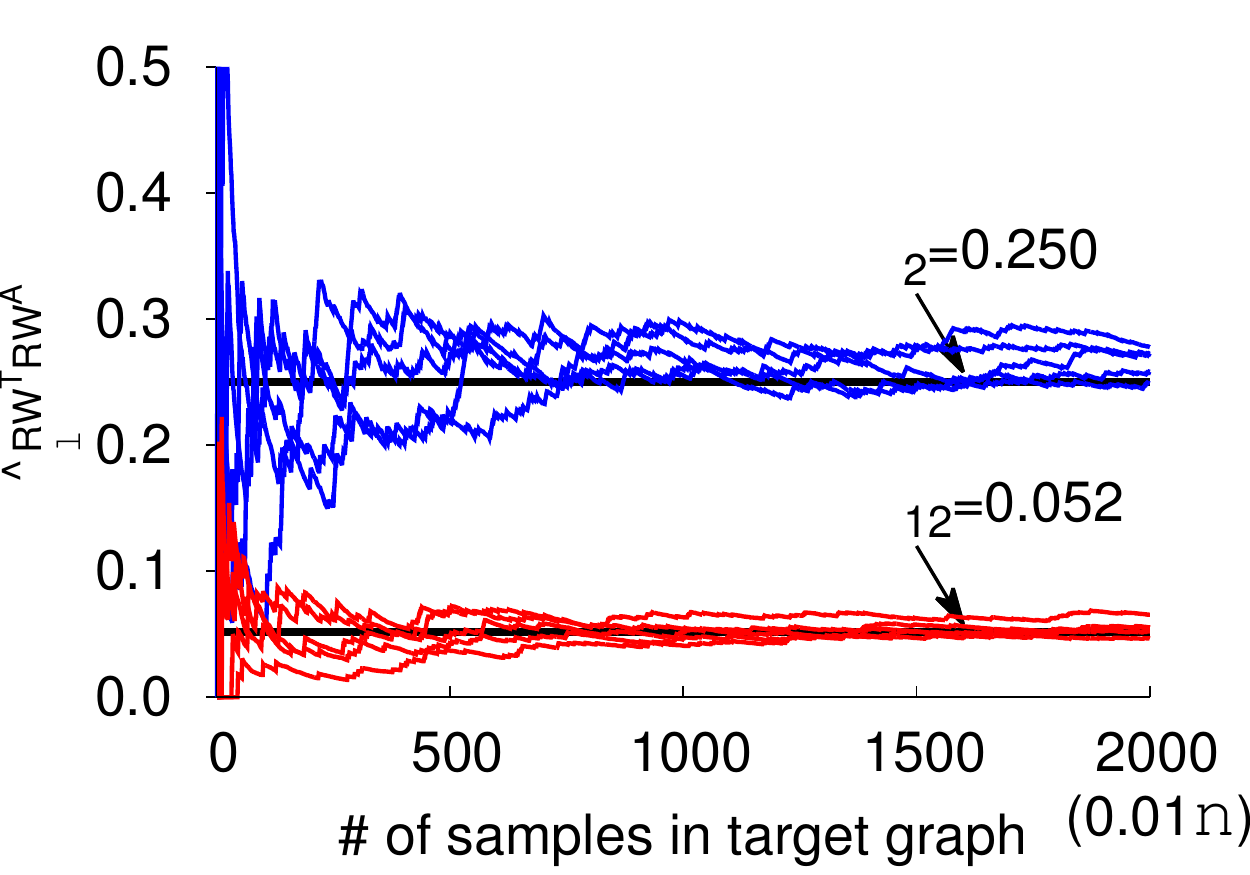}}
\caption{Asymptotic unbiasedness of the estimators ($l=2,12$).}
\label{fig:unbias}
\end{figure*}

\begin{figure*}[htp]
\centering
\subfloat[\VS{A}]{\includegraphics[width=.33\linewidth]{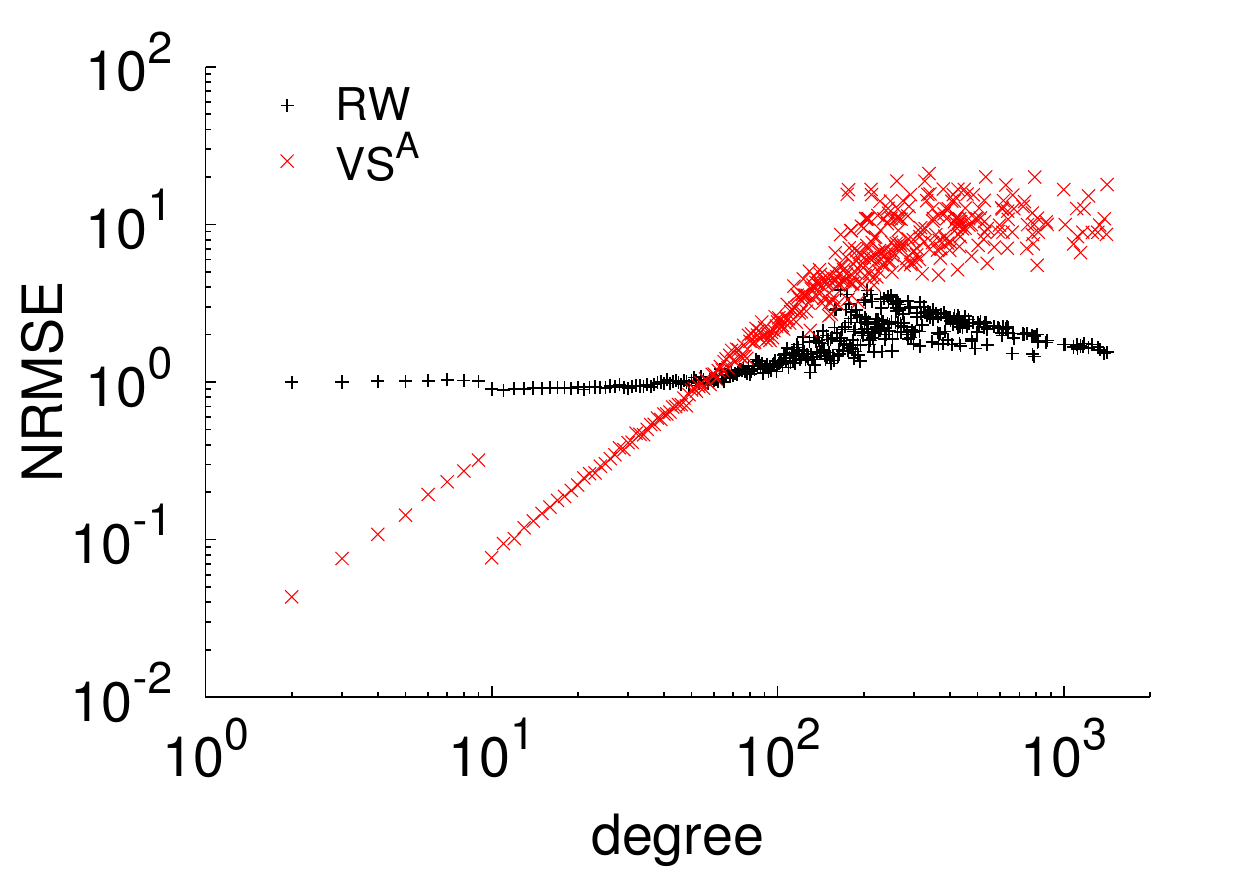}}
\subfloat[\RW{T}\VS{A}\label{fig:nrmse_rwtvsa}]{
	\includegraphics[width=.33\linewidth]{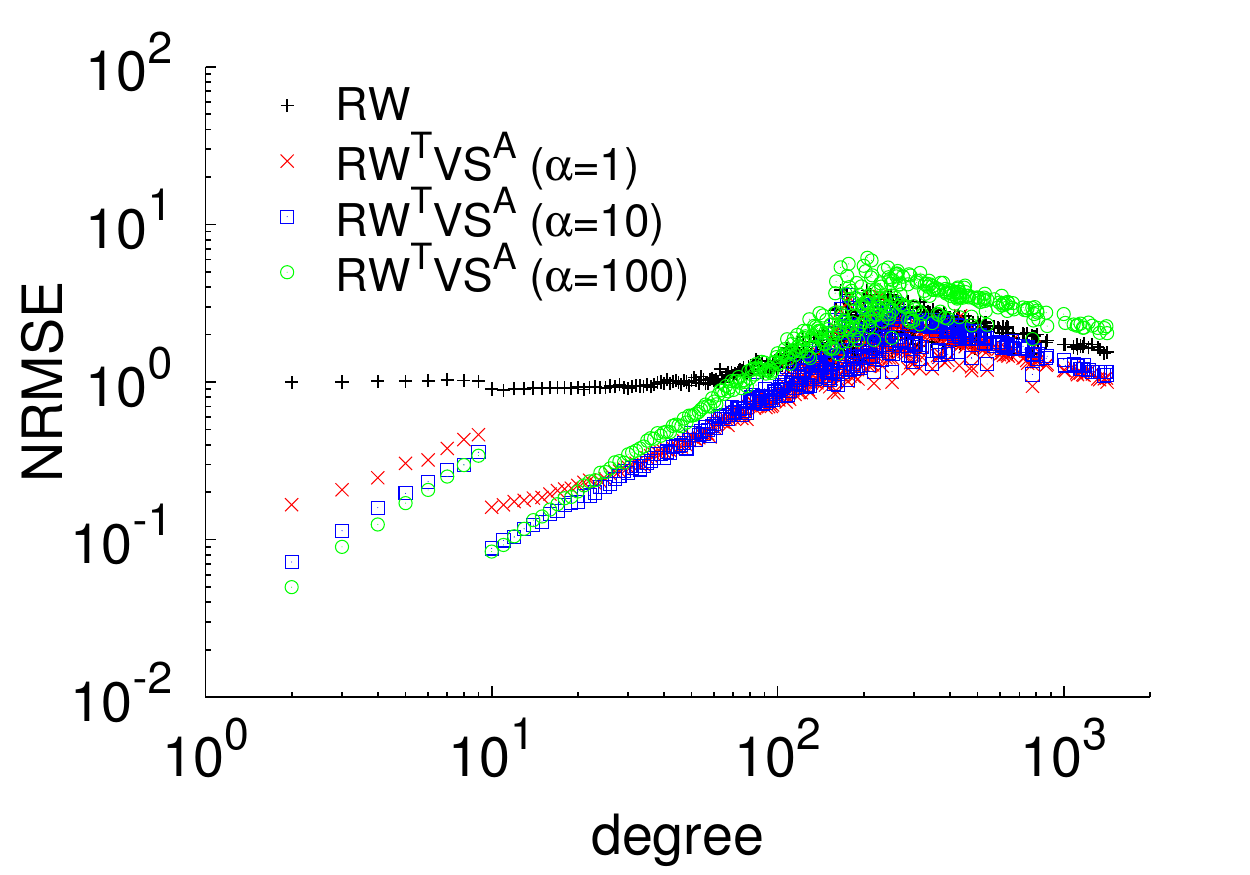}}
\subfloat[\RW{T}\RW{A}\label{fig:nrmse_rwtrwa}]{
	\includegraphics[width=.33\linewidth]{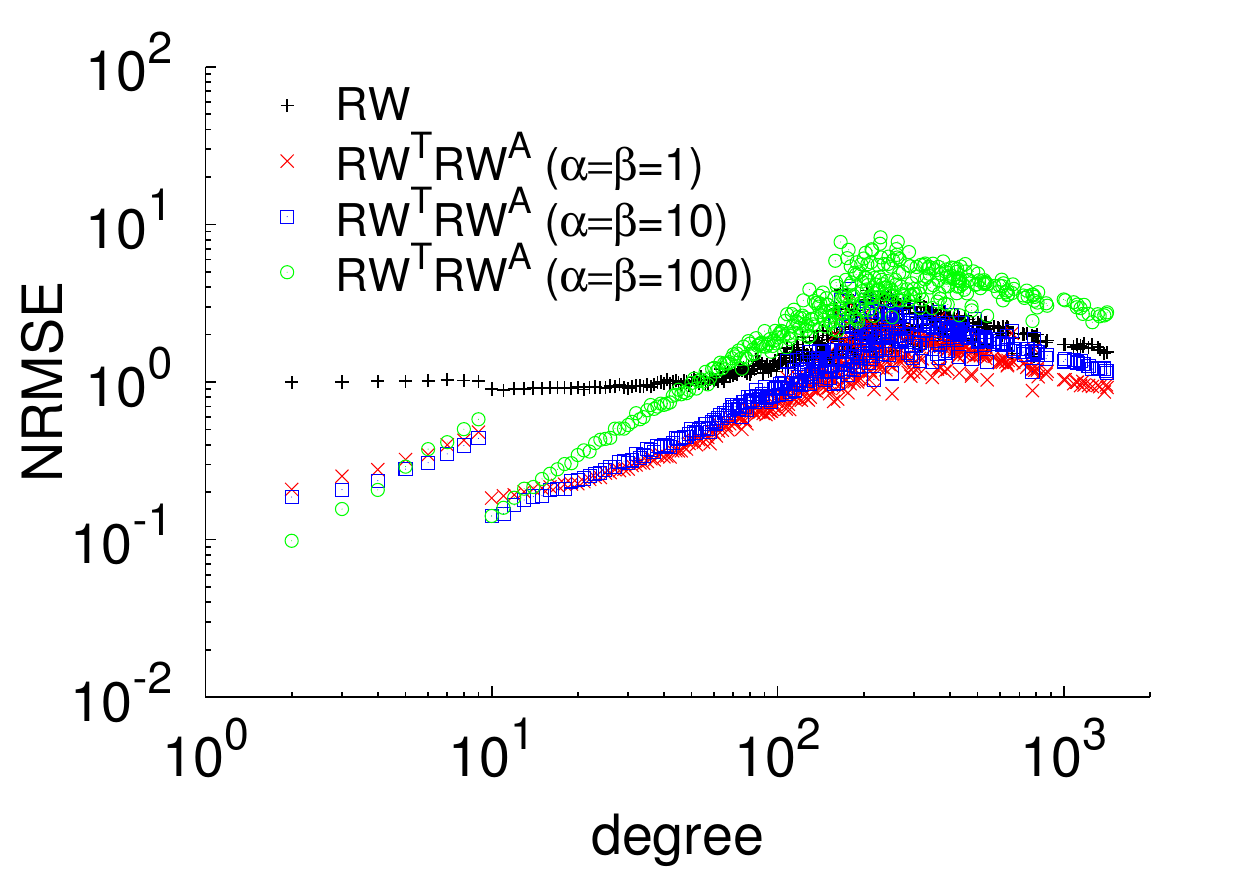}}
\caption{Estimation error comparison (averaged over $1,000$ runs, $2,000$
samples in target graph).}
\label{fig:nrmse}
\end{figure*}

\subsection{Experiments on Synthetic Data}

We first examine the soundness of the proposed sampling methods using
synthetic data.

\smallskip\noindent\textbf{Synthetic Data.}
We generate a hybrid social-affiliation network by connecting three
Barab\'{a}si-Albert graphs~\cite{Barabasi1999}, namely $G_1, G_2$ and $G_3$.
Each BA graph contains $100,000$ nodes but they have different average
degrees, i.e., $4,10$ and $20$ respectively.
$G_1$ and $G_3$ are connected with one edge to form the target graph $G$.
$G_2$ is the auxiliary graph $G'$. The affiliation graph $G_b$ is formed by
following two rules:
\begin{enumerate}
\item For each node $u$ in $G$, we connect $u$ to a randomly selected node
$v$ in $G'$;
\item We randomly choose $200,000$ pairs of nodes in $G$ and $G'$, and
connect them to form edges in $G_b$. 
\end{enumerate}
The first rule makes sure that every node in $\CU$ satisfies $d_u^{(b)}>0$.
Therefore we can apply \VS{A} method on the synthetic graph.

\smallskip\noindent\textbf{Results and Analysis.}
First we demonstrate that the proposed estimators $\check\theta_l^{\VS{A}},
\hat\theta_l^{\RW{T}\VS{A}}$ and $\hat\theta_l^{\RW{T}\RW{A}}$ are
asymptotically unbiased.
To show this, we use different sampling budgets, i.e., number of samples in
the target graph, and compare the estimator values with the ground truth.
The results are depicted in Fig.~\ref{fig:unbias}.
We use these sampling methods to estimate $\theta_2$ and $\theta_{12}$.
As the sampling budget increases, all the estimators converge to the ground
truth values, thereby demonstrating their asymptotic unbiasedness.

To compare the performance of proposed sampling methods with existing
methods, e.g., a simple random walk (RW) on $G$, we use the
\emph{normalized rooted mean squared error} (NRMSE) to evaluate the
estimation error of an estimator, which is defined as follows.
\[
\text{NRMSE}(\hat\theta_l) = 
\frac{\sqrt{\E{(\hat\theta_l-\theta_l)^2}}}{\theta_l}.
\]
The smaller the NRMSE, the better an estimator is.
We fix the sampling budget to be $1\%$ of the nodes in target graph, and
calculate the averaged empirical NRMSE over $1,000$ runs in
Fig.~\ref{fig:nrmse}.

Comparing \VS{A} with RW, we find that \VS{A} can provide smaller NRMSE for
low degree nodes than RW.
However, \VS{A} produces larger NRMSE for high degree nodes than RW.
Therefore, \VS{A} can better estimate low degree nodes in a graph
but not high degree nodes than the RW estimator.

The weakness of \VS{A} can be overcome by \RW{T}\VS{A} and \RW{T}\RW{A}.
From Figs.~\ref{fig:nrmse_rwtvsa} and~\ref{fig:nrmse_rwtrwa} we can see
that, when we allow jumps in \RW{T}\VS{A} and \RW{T}\RW{A} by setting
$\alpha=\beta=1$, the NRMSE for high degree nodes decreases as
small as RW, and NRMSE for low degree nodes keeps smaller than RW.
If we increase the probability of jumping at each step by increasing
$\alpha$ and $\beta$, we observe that the NRMSE for low degree nodes
decreases, and NRMSE for high degree nodes increases. 
This behavior is similar to RWwJ~\cite{Avrachenkov2010,Ribeiro2012b}
because \RW{T}\VS{A} and \RW{T}\RW{A} are equivalent to RWwJ in the design.

\subsection{Experiments on LBSN Datasets}

Now we conduct experiments on two real-world location-based social network
(LBSN) datasets to solve the problem mentioned in Example~\ref{exam:weibo}.

\smallskip\noindent\textbf{LBSN Datasets.}
We use two public datasets Brightkite and Gowalla~\cite{Cho2011} to solve 
our first problem in Example~\ref{exam:weibo}.
Brightkite and Gowalla are once two popular LBSNs where users shared their
locations by checking-in.
Users are also connected by undirected friendship relationships.
The statistics of the two datasets are summarized in Table~\ref{tab:lbsn}.

\begin{table}[htp]
\centering
\caption{Summary of LBSN datasets. \label{tab:lbsn}}
\begin{threeparttable}[b]
\begin{tabular}{|c|l|r|r|}
\hline
\multicolumn{2}{|c|}{dataset} & Brightkite & Gowalla \\
\hline
\hline
\multirow{5}{*}{$G$}
 & network type & undirected & undirected \\
 & \# of users & $58,228$ & $196,591$ \\
 & \# of friendship edges & $214,078$ & $950,327$ \\
 & \# of users in LCC\tnote{1} & $56,739$ & $196,591$ \\
 & \# of edges in LCC & $212,945$ & $950,327$ \\
\hline
\multirow{3}{*}{$G'$ and $G_b$} 
 & \# of distinct venues & $772,966$ & $1,280,969$ \\
 & \# of users having check-ins & $51,406$ & $107,092$ \\
 & \# of check-ins & $4,491,143$ & $6,442,890$ \\
\hline
\multirow{3}{*}{\minitab{$G'$ and $G_b$ \\ for NYC}} 
 & \# of venues in NYC\tnote{2} & $23,484$ & $26,448$ \\
 & \# of users checking in NYC & $4,257$ & $7,399$ \\
 & \# of check-ins in NYC & $33,656$ & $113,423$ \\
\hline
\end{tabular}
\begin{tablenotes}
\item[1]The largest connected component.
\item[2]The New York City (Fig.~\ref{fig:nyc}).
\end{tablenotes}
\end{threeparttable}
\end{table}

\smallskip\noindent\textbf{Venue Sampling.} 
Suppose the surveyor wants to measure characteristics of users located
in New York City (NYC, latitude $40.4^\circ\sim 41.4^\circ$, longitude
$-74.3^\circ\sim -73.3^\circ$, see Fig.~\ref{fig:nyc}), i.e., the degree
distribution of users who checked in NYC.
As we explained in Introduction, directly sampling users is not a good
idea.
Here, we apply the \VS{A} method along with a venue sampling method Random
Region Zoom-In (RRZI)~\cite{Wang2014} to sample users in NYC more
efficiently.

\begin{figure}[htp]
\centering
\includegraphics[width=.8\linewidth]{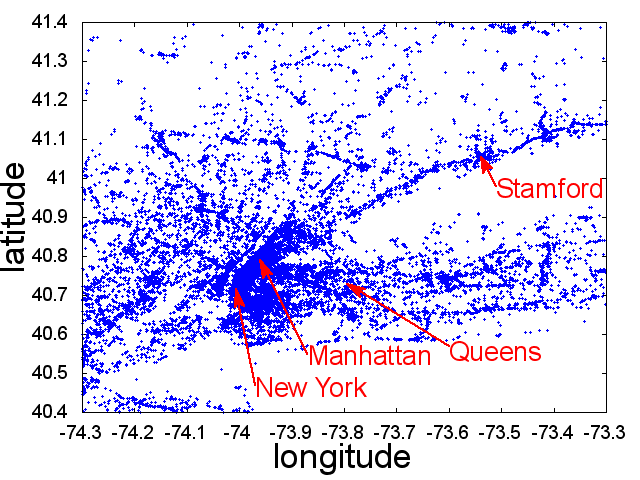}
\caption{Venue distribution in New York City}
\label{fig:nyc}
\end{figure}

RRZI utilizes the venue query API provided by most LBSNs.
A user first specifies a region by giving the bottom-left and top-right
corner latitude-longitude coordinates, and the API returns a set of venues
in this region.
Usually, the size of the set returned is limited to at most $K$ if there
are more than $K$ venues in the region.
RRZI regularly zooms in the region until the subregion is fully accessible,
i.e., the API returns less than $K$ venues in the subregion.
Therefore, RRZI can provide samples of venues in a region of interest.

\smallskip\noindent\textbf{Results.} 
Combining the RRZI and \VS{A} methods, we conduct two experiments to
indirectly sample users in NYC on Brightkite and Gowalla respectively.
We sample $1\%$ of venues in NYC and calculate the degree distribution of
users in NYC.
The results are depicted in Figs.~\ref{fig:bri_rrzi_vsa}
and~\ref{fig:gow_rrzi_vsa}.

\begin{figure}[htp]
\centering
\subfloat[RRZI-\VS{A} estimates\label{fig:bri_est}]{
	\includegraphics[width=.5\linewidth]{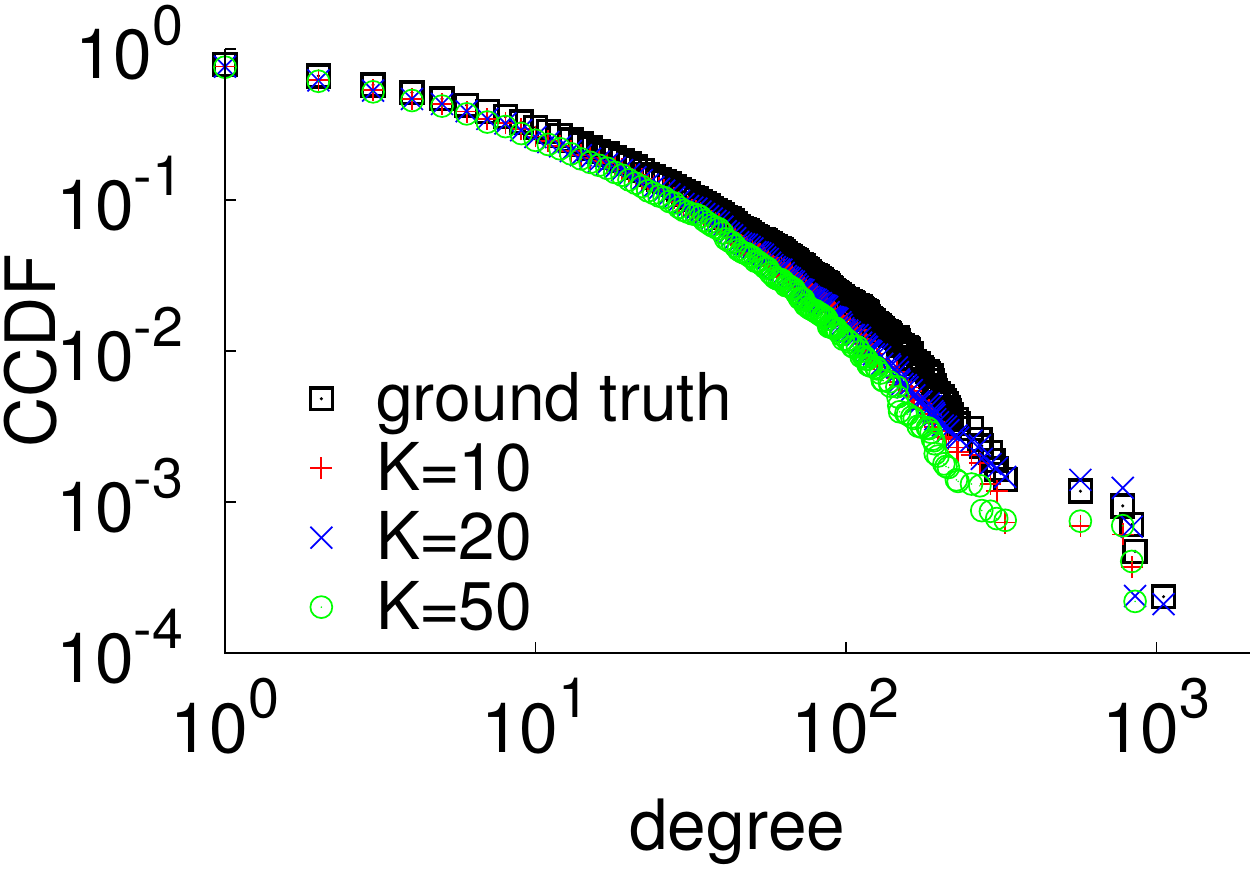}}
\subfloat[RRZI-\VS{A} NRMSE\label{fig:bri_nrmse}]{
	\includegraphics[width=.5\linewidth]{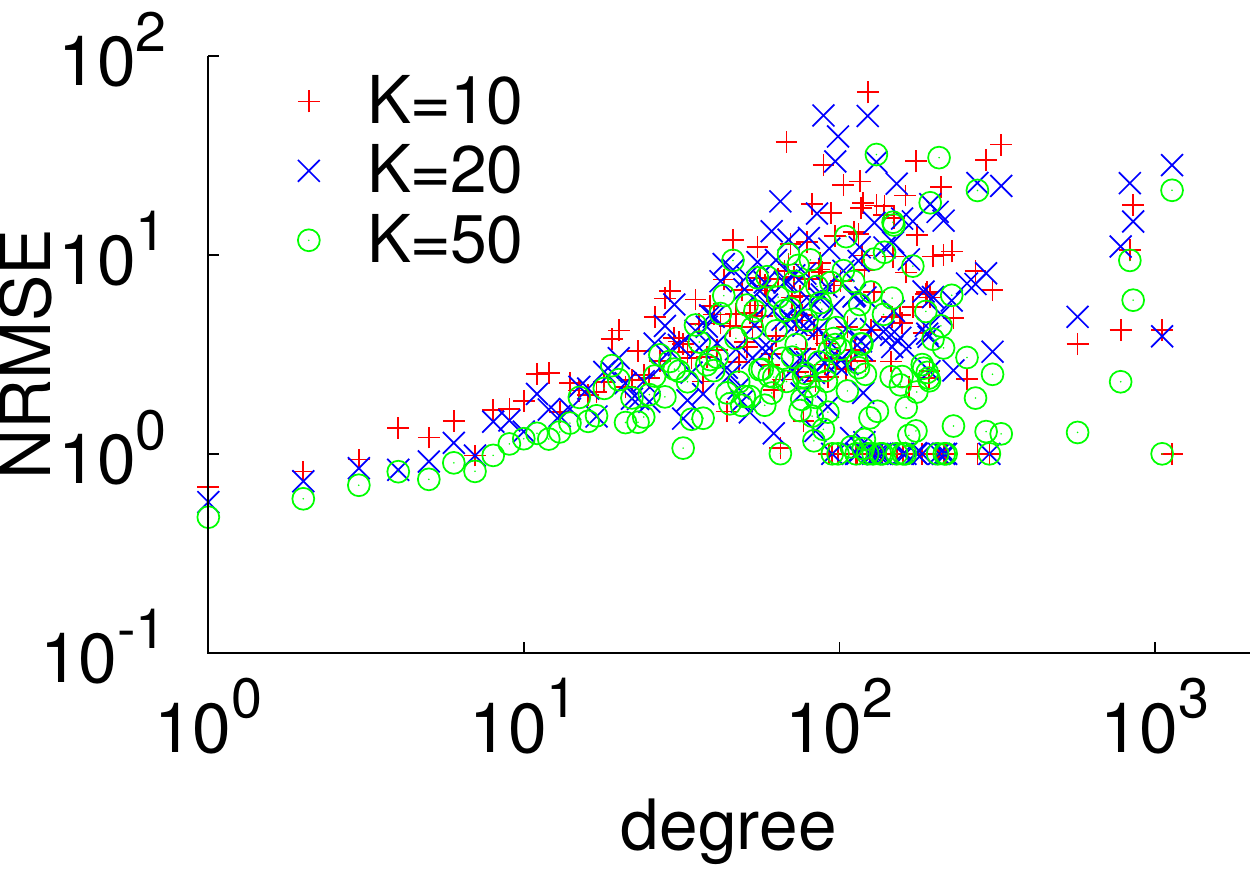}}
\caption{User characteristics in NYC on Brightkite. (averaged over $1000$
runs)}
\label{fig:bri_rrzi_vsa}
\end{figure}

\begin{figure}[htp]
\centering
\subfloat[RRZI-\VS{A} estimates\label{fig:gow_est}]{
	\includegraphics[width=.5\linewidth]{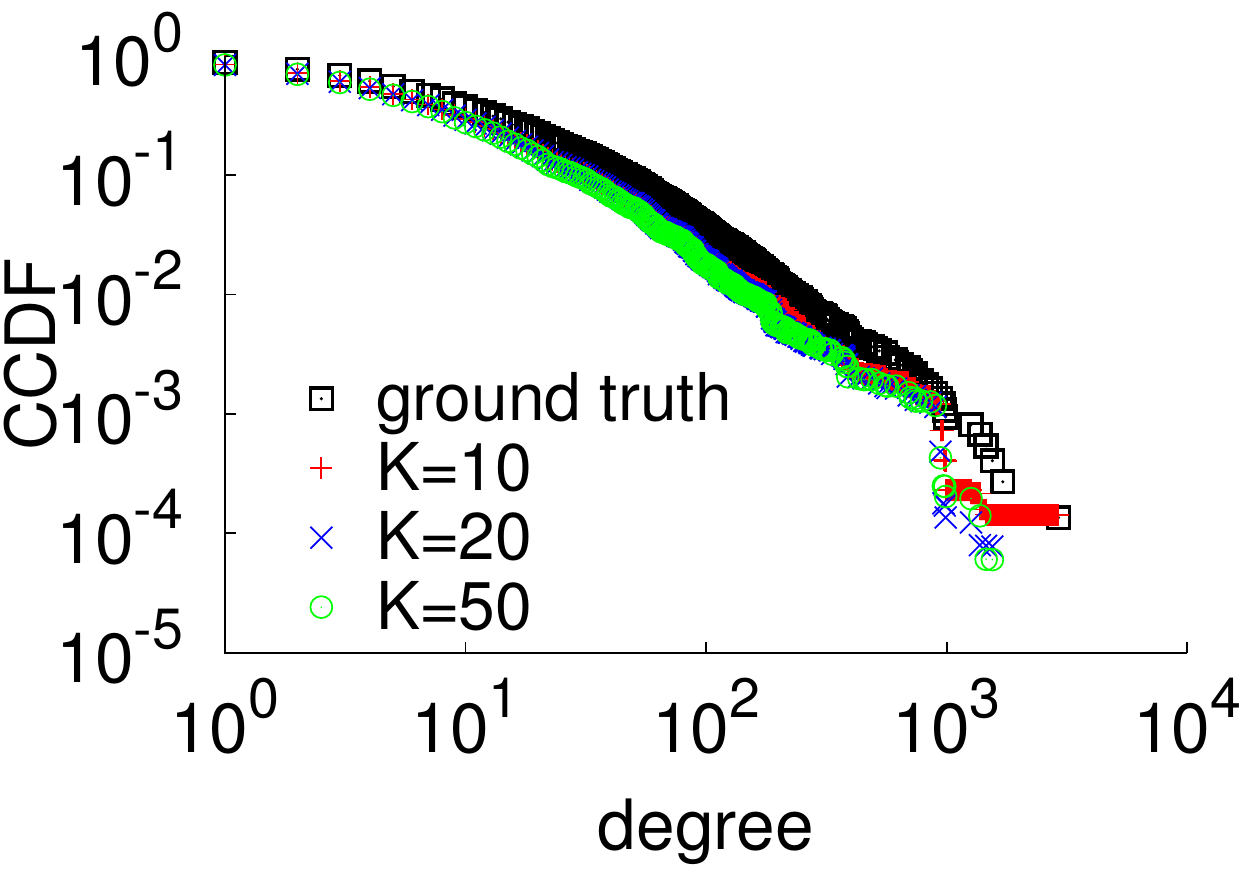}}
\subfloat[RRZI-\VS{A} NRMSE\label{fig:gow_nrmse}]{
	\includegraphics[width=.5\linewidth]{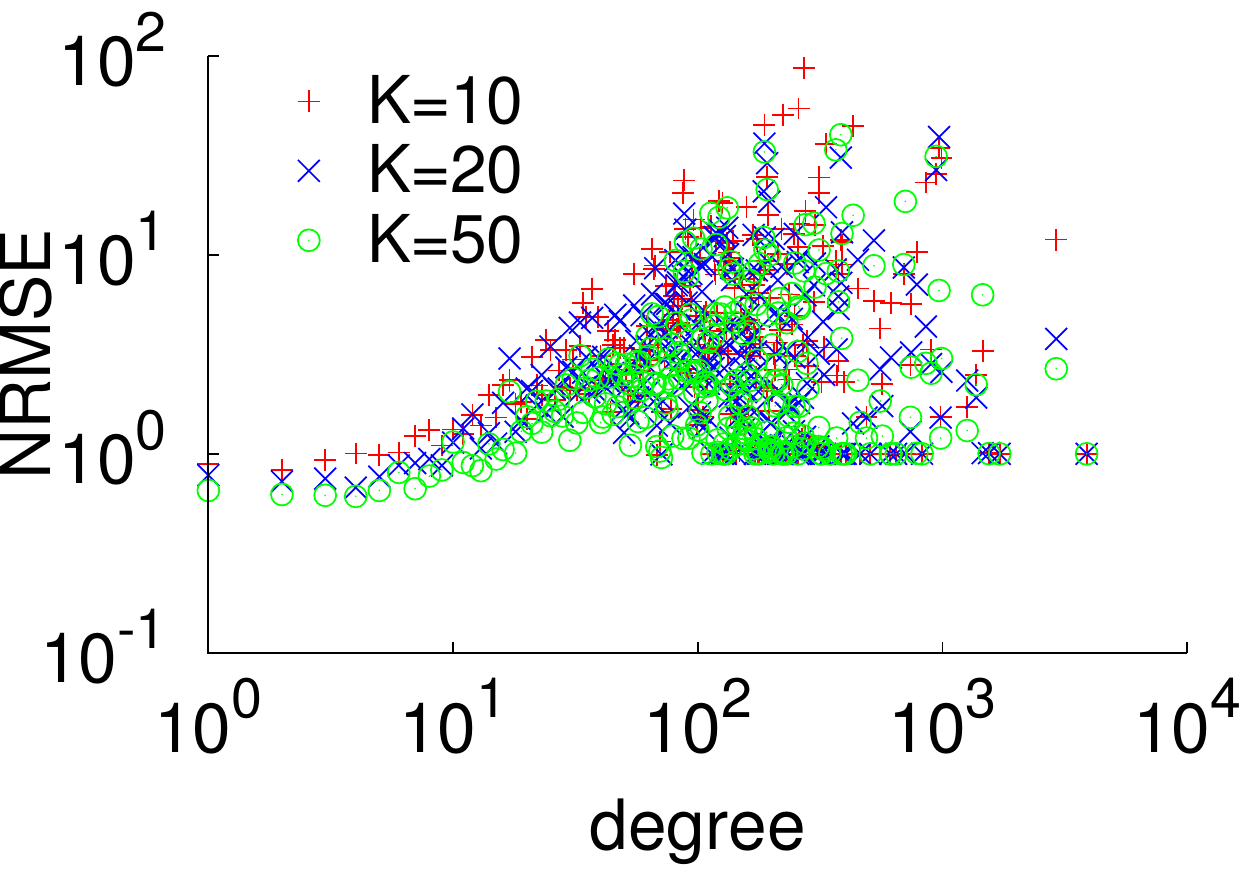}}
\caption{User characteristics in NYC on Gowalla. (averaged over $1000$
runs)}
\label{fig:gow_rrzi_vsa}
\end{figure}

From Figs.~\ref{fig:bri_est} and~\ref{fig:gow_est}, we can see that
RRZI-\VS{A} method can provide well estimates of users in NYC,
and the estimates for low degree users are better than high degree users,
which are clearer from Figs.~\ref{fig:bri_nrmse} and~\ref{fig:gow_nrmse}.
These results are consistent with our previous analysis on synthetic data.
In fact, we can combine \VS{A} with other venue sampling
methods~\cite{Li2012,Li2014,Wang2014} to provide better estimates than RRZI.
However, we omit them due to space limitation.

\subsection{Experiments on Mtime Dataset}
Next, we conduct experiments on Mtime to measure Mtime user characteristics
in Example~\ref{exam:mtime}.

\smallskip\noindent\textbf{Mtime Dataset.}
As we have introduced in Example~\ref{exam:mtime}, users and actors in 
Mtime naturally form a hybrid social-affiliation network.
To build a ground-truth dataset, we almost completely collected the Mtime
users and actors data by traversing user IDs ranging from $100000$ to
$10000000$, and actor IDs ranging from $892000$ to $2100000$.\footnote{
Moreover, Mtime does not restrict HTTP request frequency from third parties.
So we can finish the data collecting within one week using eight machines.
}

For each Mtime user, we collect the set of users he following and users
following him. 
This builds up a directed follower network among users in Mtime.
We also collect the profile information for each user, including gender,
location, tags, groups and so on.
Moreover, each user maintains a list including actors that interest him.
This forms a fan-relationship between users and actors.
For each Mtime actor, we collect the movies he/she participated in.
This can build up a cooperative network among actors, e.g., if two actors
participated in a same movie, they have a cooperative relationship between
them.
The complete Mtime dataset is summarized in Table~\ref{tab:mtime}.

\begin{table}[htbp]
\centering
\caption{Summary of Mtime dataset. \label{tab:mtime}}
\begin{threeparttable}[b]
\begin{tabular}{|c|l|r|}
\hline
\multirow{6}{*}{$G$}
 & user follower network type & directed \\
 & total users (isolated or non-isolated)\tnote{3} & $1,878,127$ \\
 & \# of non-isolated users in follower network & $1,035,164$ \\
 & \# of following relationships & $14,861,383$ \\
 & \# of users in LCC & $987,055$ \\
 & \# of following relationships in LCC & $14,791,482$ \\
\hline
\multirow{6}{*}{$G'$}
 & actor cooperative network type & undirected \\
 & total actors (isolated or non-isolated) & $1,123,340$ \\
 & \# of non-isolated actors in cooperative network & $1,122,166$ \\
 & \# of cooperative relationships & $10,344,364$ \\
 & \# of actors in LCC & $1,114,065$ \\
 & \# of cooperative relationships in LCC & $10,328,904$ \\
\hline
\multirow{7}{*}{$G_b$}
 & \# of fan relationships & $225,558,343$ \\
 & \# of users following actors & $1,419,339$ \\
 & \# of isolated users following actors & $842,963$ \\
 & \# of actors having fans & $441,413$ \\
 & \# of isolated actors having fans & $1,174$ \\
 & \# of isolated actors having only isolated fans & $225$ \\
 & \# of isolated users following only isolated actors & $393$ \\
\hline
\end{tabular}
\begin{tablenotes}
\item[3] An isolated node in a graph is a node with zero degree.
\end{tablenotes}
\end{threeparttable}
\end{table}

\smallskip\noindent\textbf{Analysis of the Dataset.}
First, we provide some basic analysis of the Mtime dataset.
In Table~\ref{tab:mtime}, we compare the first and second blocks, which are
related to the target graph $G$ and auxiliary graph $G'$ respectively.
We find that about $14\%$ of the user IDs and $93\%$ of the actor IDs are
valid.
This indicates that vertex sampling in auxiliary graph is more efficient
than in target graph.
Moreover, we can find that about $45\%$ of users are isolated, i.e., having
zero degree, but the same number for actors is less than $0.1\%$.
This indicates that the auxiliary graph is better connected than the
target graph.
Although a large fraction of users are isolated nodes in the target graph,
from the last block in Table~\ref{tab:mtime} (regarding the affiliation
graph $G_b$), we find that almost all the isolated users are connected to
non-isolated actors (except a few hundreds of them).
So the majority of isolated users are indirectly connected to other users
through actors.
This is illustrated in Fig.~\ref{fig:hsa_mtime}.
By introducing the hybrid social-affiliation network, we can study a larger
user sample space than the largest connected component of user graph.

\begin{figure}[htp]
\centering
\begin{tikzpicture}[
und/.style={draw,thick,circle,minimum size=6pt,inner sep=0},
vnd/.style={draw=blue,thick,rectangle,minimum size=6pt,inner sep=0},
att/.style={left,fill=white,minimum size=0,inner sep=0,align=center}]

\node[und] (u11) at (0,0) {};
\node[und,above right = 0.5 and 0.3 of u11] (u12) {};
\node[und,above right = 0.4 and 1.0 of u11] (u13) {};
\node[und,right = 1 of u11] (u14) {};
\node[und,above right = 0.1 and 0.5 of u11] (u15) {};
\node[und,yshift=0.5cm] (u16) at (u11) {};
\draw[thick] (u16)--(u11)--(u12)--(u13)--(u14)--(u11)--(u15)--(u12) (u15)--(u14);
\node[draw,dashed,fit={(u11) (u12) (u13) (u14) (u15)}] (fu1) {};
\node[att,below=0.1 of fu1] {LCC of $G$\\ ($987,055$ users)};

\node[und,right =3 of u11] (u21) {};
\node[und,xshift=0.1cm,yshift=0.5cm] (u22) at (u21) {};
\node[und,above right = 0.5 and 1.0 of u21] (u23) {};
\node[und,right = 1 of u21] (u24) {};
\node[und,xshift=0.5cm,yshift=0.65cm] (u25) at (u21) {};
\node[und,xshift=0.5cm,yshift=0.1cm] (u26) at (u21) {};
\node[und,xshift=0.3cm,yshift=0.3cm] (u27) at (u26) {};
\node[draw,dashed,fit={(u21) (u22) (u23) (u24) (u25)}] (fu2) {};
\draw[thick] (u22)--(u25) (u26)--(u27);
\node[att,below=0.1 of fu2] {Isolated parts of $G$\\ ($858,733$ users)};

\node[vnd,yshift=1.9cm] (v11) at (u11) {};
\node[vnd,above right = 0.5 and 0.3 of v11] (v12) {};
\node[vnd,above right = 0.4 and 1.0 of v11] (v13) {};
\node[vnd,right = 1 of v11] (v14) {};
\node[vnd,above right = 0.1 and 0.5 of v11] (v15) {};
\node[vnd,yshift=0.55cm] (v16) at (v11) {};
\draw[thick,blue] (v11)--(v12)--(v13)--(v14)--(v15)--(v13) (v11)--(v14) (v15)--(v12)--(v16);
\node[draw,dashed,fit={(v11) (v12) (v13) (v14) (v15)}] (fa1) {};
\node[att,above=0.1 of fa1] {LCC of $G'$\\ ($1,114,065$ actors)};

\node[vnd,right =3 of v11] (v21) {};
\node[vnd,xshift=0.1cm,yshift=0.5cm] (v22) at (v21) {};
\node[vnd,above right = 0.5 and 1 of v21] (v23) {};
\node[vnd,right = 1 of v21] (v24) {};
\node[vnd,xshift=0.5cm,yshift=0.65cm] (v25) at (v21) {};
\node[vnd,xshift=0.5cm,yshift=0.1cm] (v26) at (v21) {};
\node[vnd,xshift=0.4cm,yshift=0.3cm] (v27) at (v26) {};
\node[draw,dashed,fit={(v21) (v22) (v23) (v24) (v25)}] (fa2) {};
\draw[thick,blue] (v27)--(v25) (v21)--(v26);
\node[att,above=0.1 of fa2] {Isolated parts of $G'$\\ ($1,658$ actors)};

\draw[line width=3pt,red,dashed] (fu1)--(fa1)--(fu2) (fu1)--(fa2)--(fu2);
\end{tikzpicture}

\caption{Hybrid social-affiliation network in Mtime. Dashed red lines
denote fan-relationships.}
\label{fig:hsa_mtime}
\end{figure}
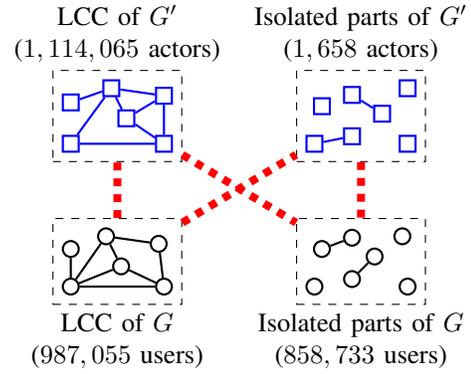

\noindent\textbf{Results.} 
Using the Mtime dataset, we demonstrate that \RW{T}\VS{A} and \RW{T}\RW{A}
methods can provide well estimates of user
characteristics.\footnote{Because not every user follows actors, we cannot
apply \VS{A} method on Mtime dataset.}
Although the user follower network is directed, we can build an undirected
version of target graph on-the-fly while sampling because a user's
in-coming and out-going neighbors are known once the user is
queried~\cite{Ribeiro2010}.

Results of method \RW{T}\VS{A} are depicted in Fig.~\ref{fig:mt_rwtvsa}.
From Figs.~\ref{fig:mt_rwtvsa}(a) and~(b), we observe
that \RW{T}\VS{A} can provide well estimates of in-degree and out-degree
distributions of the target graph.
From Figs.~\ref{fig:mt_rwtvsa}(c) and~(d), we observe
that when more nodes of the target graph are sampled, the estimation
accuracy increases (NRMSE decreases).
When more jumps are allowed by increasing $\alpha$ from $0.1$ to $10$, we
observe that the estimation accuracy of low degree nodes is increased from
Figs.~\ref{fig:mt_rwtvsa}(e) and~(f). 
This is consistent with the results on synthetic data.

\begin{figure}
\centering
\subfloat[In-degree estimates]{
	\includegraphics[width=.5\linewidth]{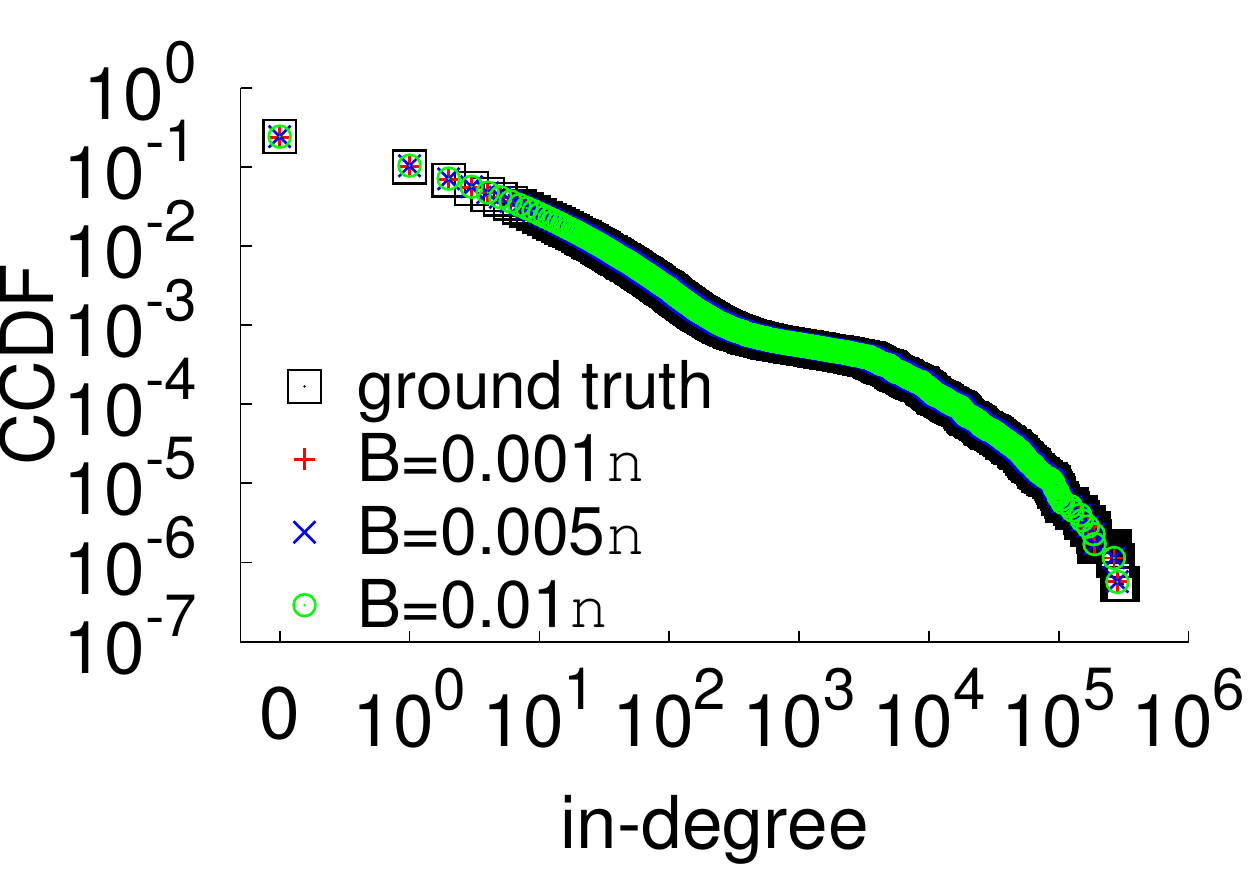}}
\subfloat[Out-degree estimates]{
	\includegraphics[width=.5\linewidth]{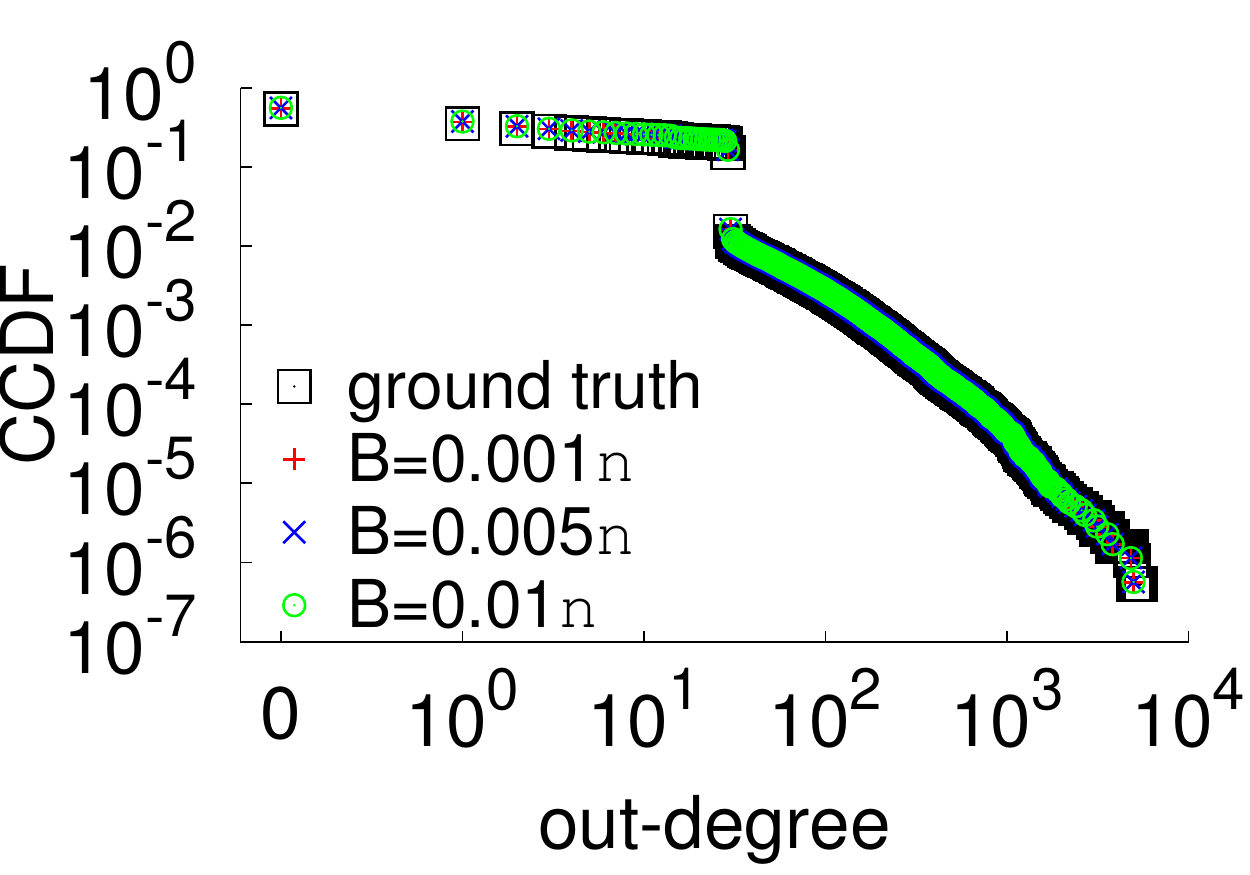}} \\
\subfloat[In-degree NRMSE]{
	\includegraphics[width=.5\linewidth]{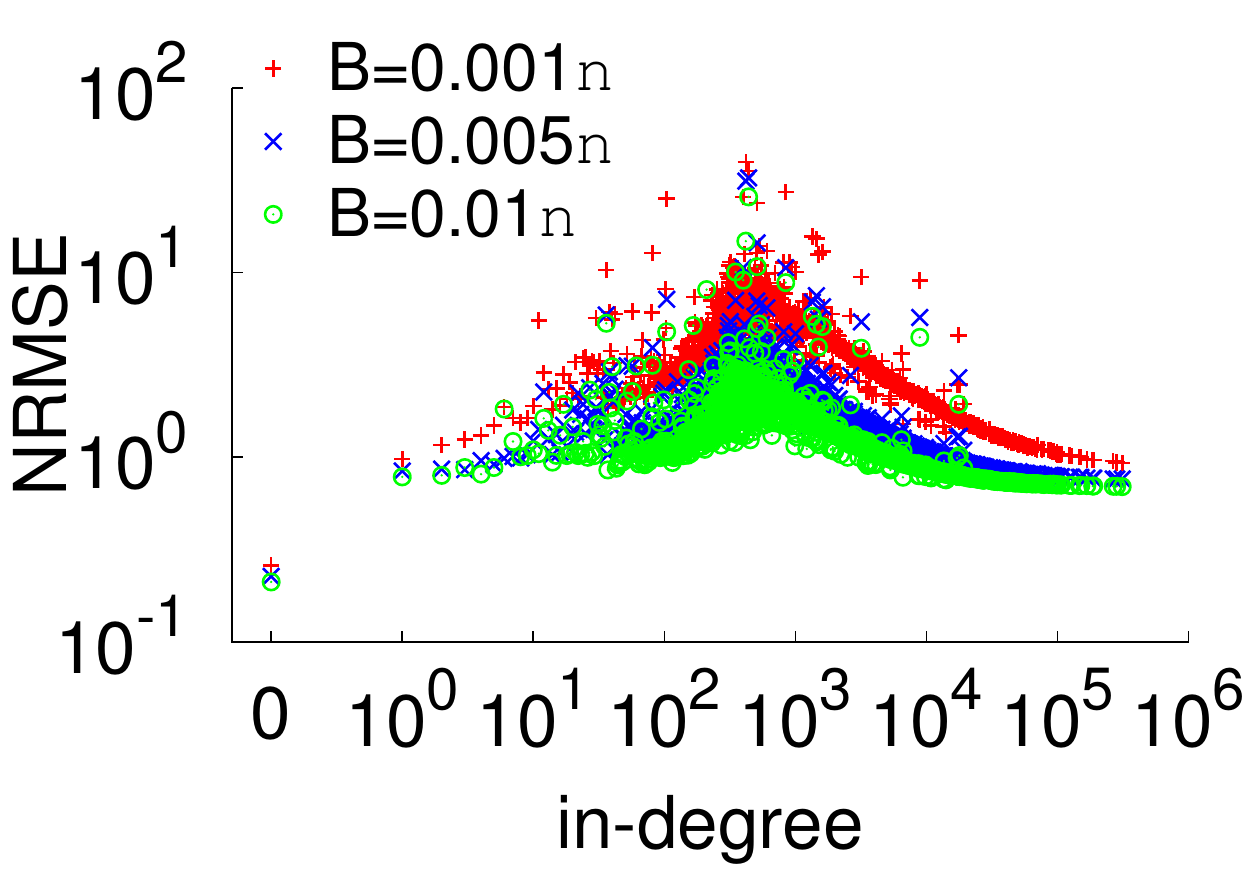}} 
\subfloat[Out-degree NRMSE]{
	\includegraphics[width=.5\linewidth]{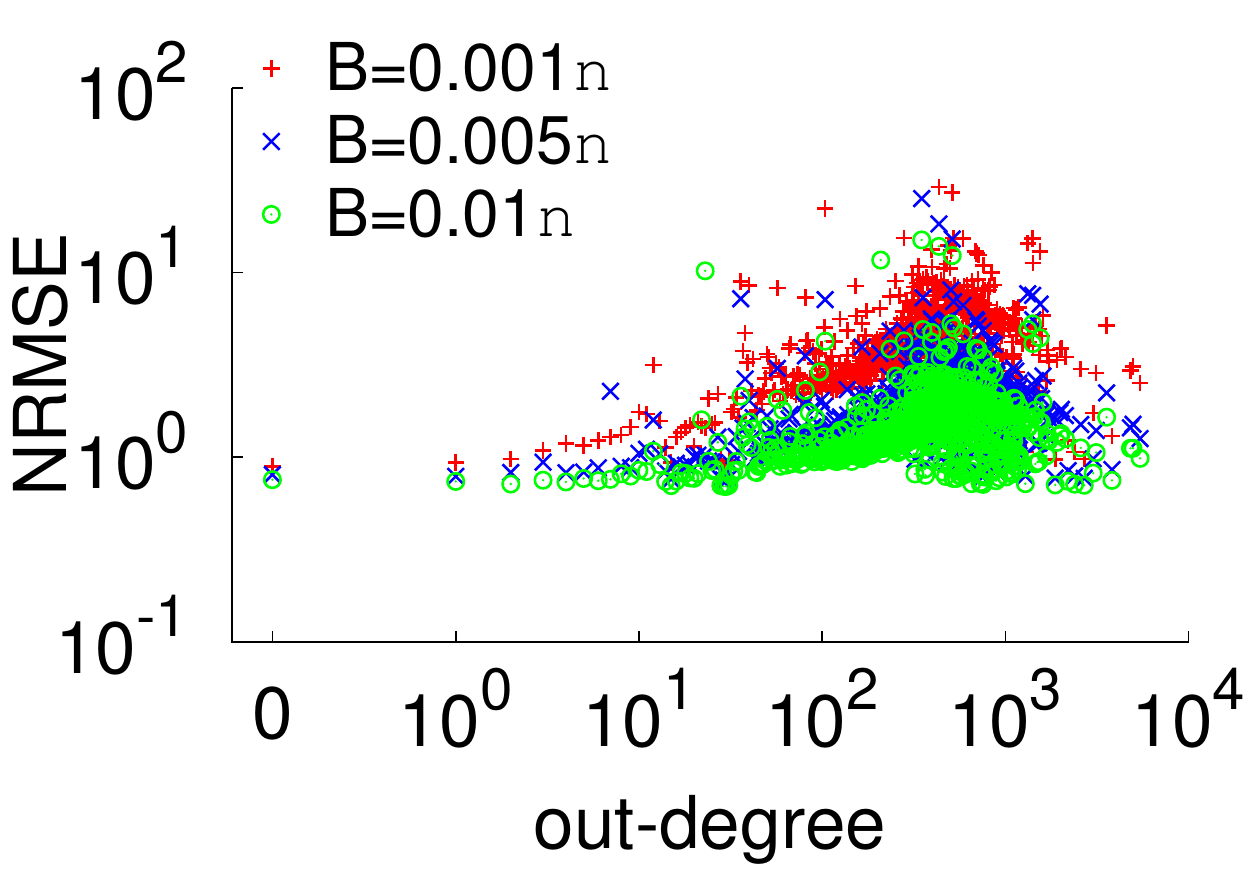}} \\
\subfloat[In-degree NRMSE]{
	\includegraphics[width=.5\linewidth]{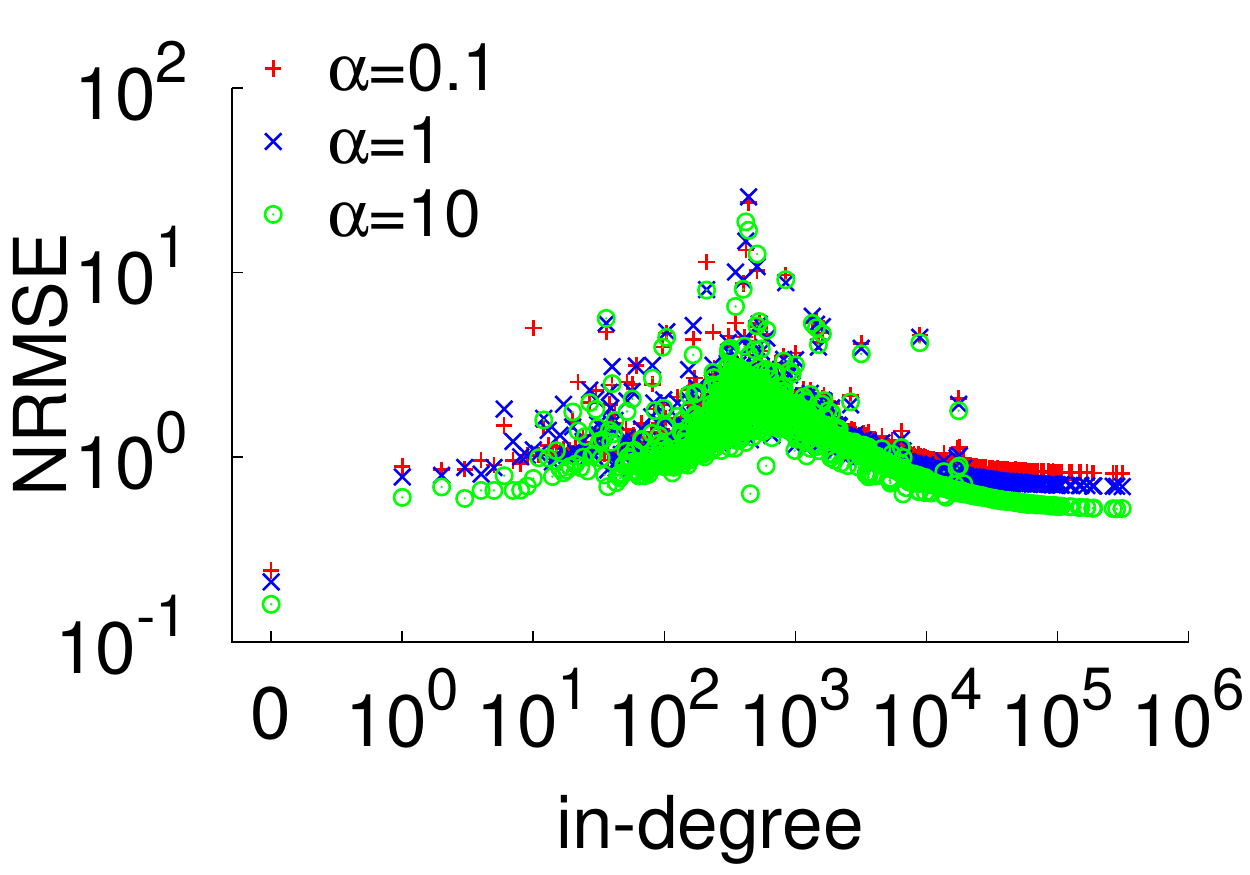}} 
\subfloat[Out-degree NRMSE]{
	\includegraphics[width=.5\linewidth]{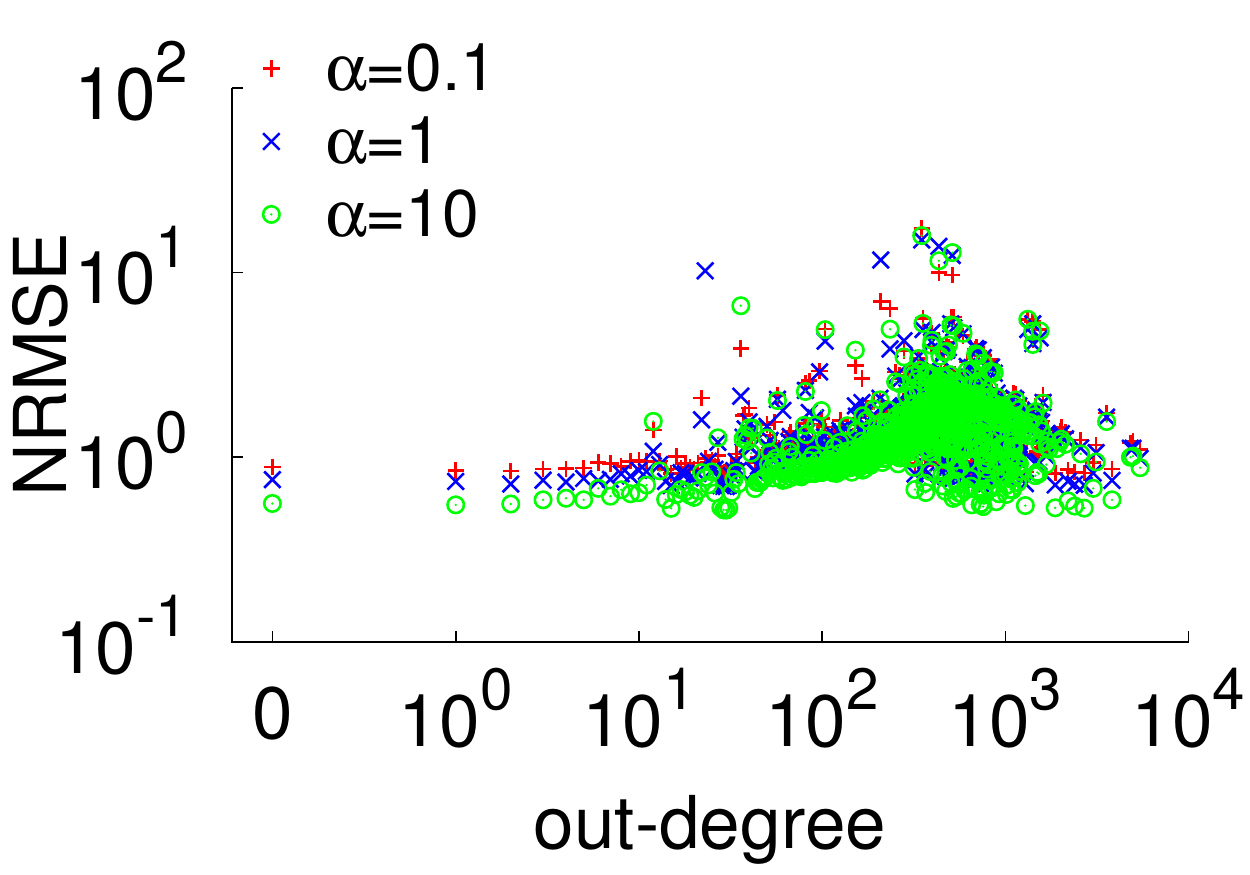}}
\caption{\RW{T}\VS{A} estimates and NRMSE. (We use $\alpha=1$ in
(a)-(d), and $B=0.01n$ in (e)-(f). Each result is averaged over $10,000$
runs.)}
\label{fig:mt_rwtvsa}
\end{figure}

Results of method \RW{T}\RW{A} are similar to the results of \RW{T}\VS{A},
and we show them in Fig~\ref{fig:mt_rwtrwa}.
First, from Figs.~\ref{fig:mt_rwtrwa}(a) and~(b) we observe that
\RW{T}\RW{A} can also provide well estimates of user characteristics.
Second, from Figs.~\ref{fig:mt_rwtrwa}(c) and~(d) we can find that when
more nodes of the target graph are sampled, NRMSE decreases thereby
increasing estimation accuracy.
Last, from Figs~\ref{fig:mt_rwtrwa}(e) and~(f) we find that when more jumps
are allowed (by increasing $\alpha$ and $\beta$), NRMSE for low degree
nodes decreases for both in-degree and out-degree estimates.

\begin{figure}
\centering
\subfloat[In-degree estimates]{
	\includegraphics[width=.5\linewidth]{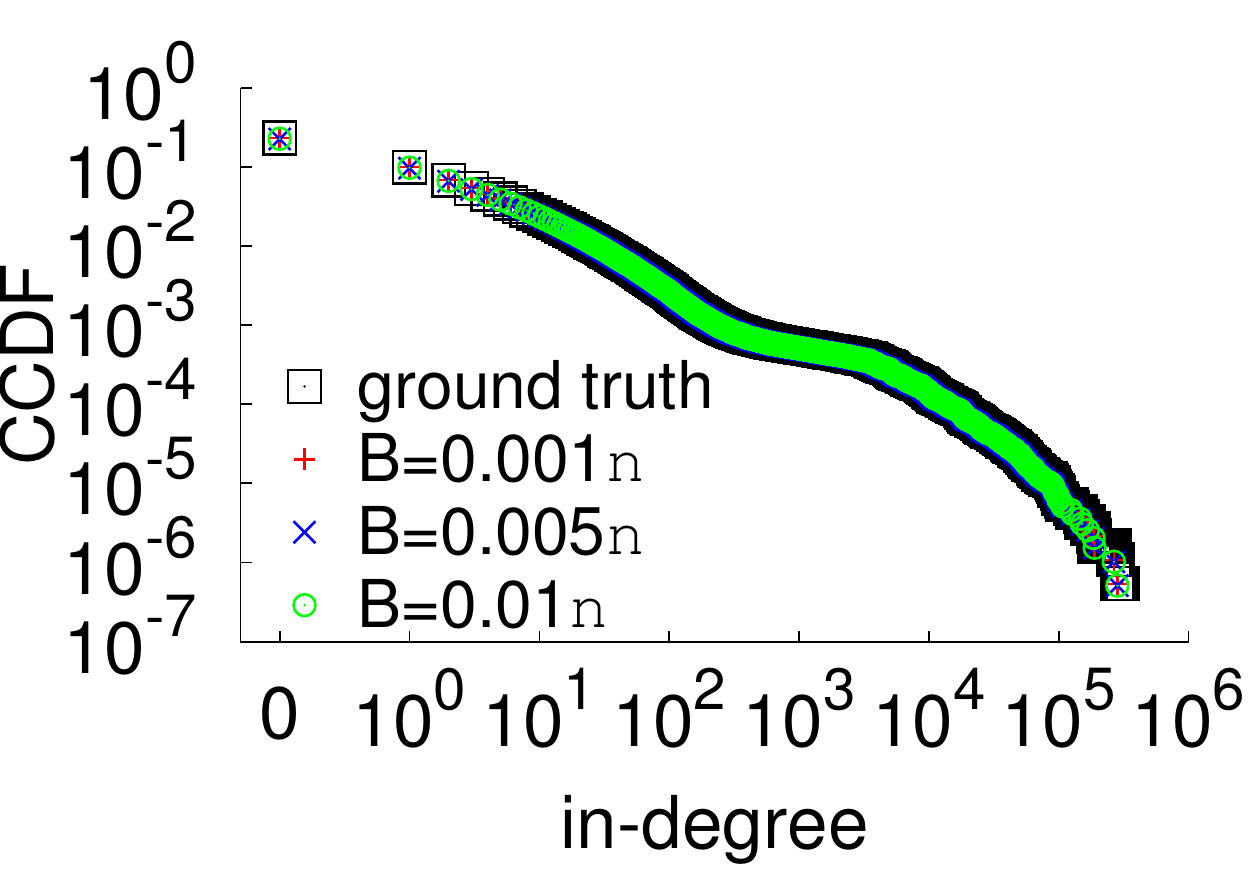}}
\subfloat[Out-degree estimates]{
	\includegraphics[width=.5\linewidth]{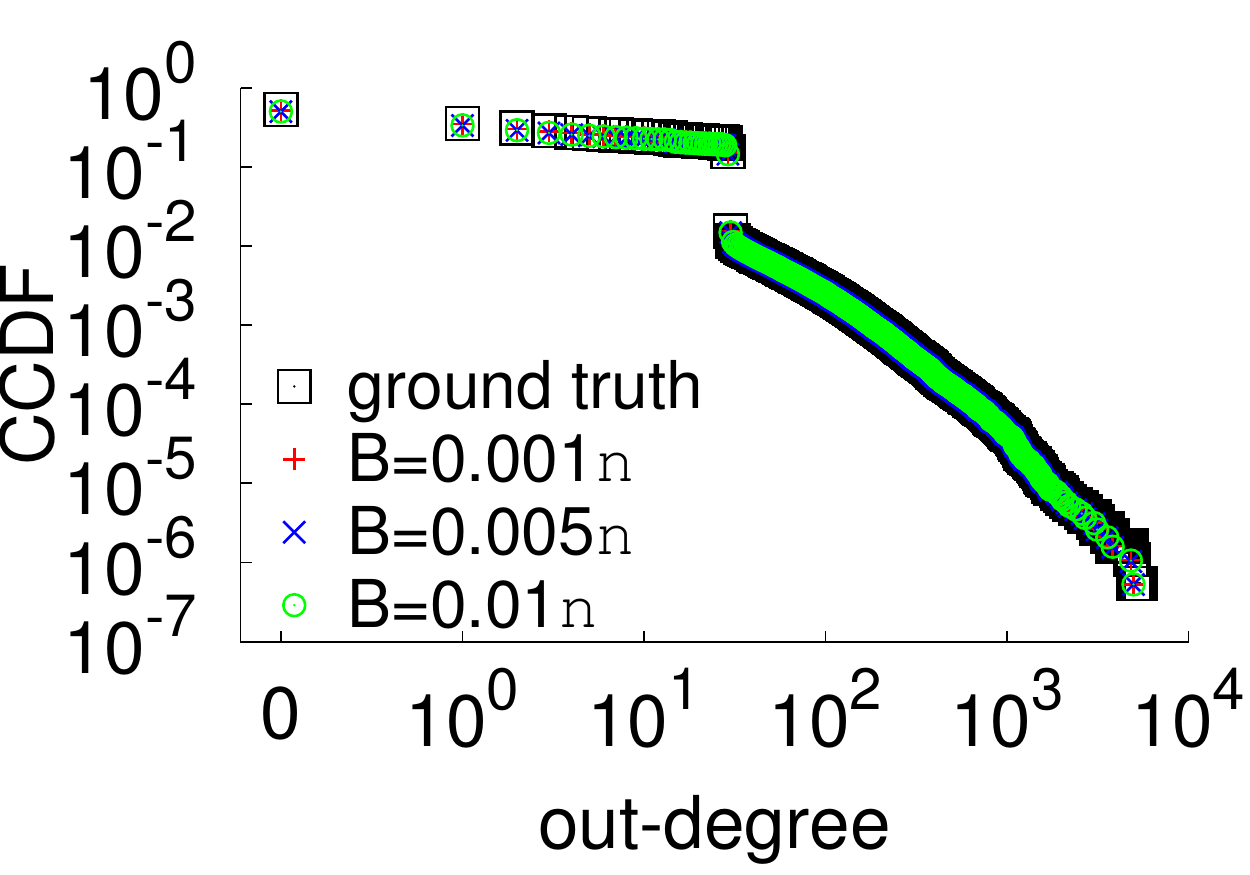}} \\
\subfloat[In-degree NRMSE]{
	\includegraphics[width=.5\linewidth]{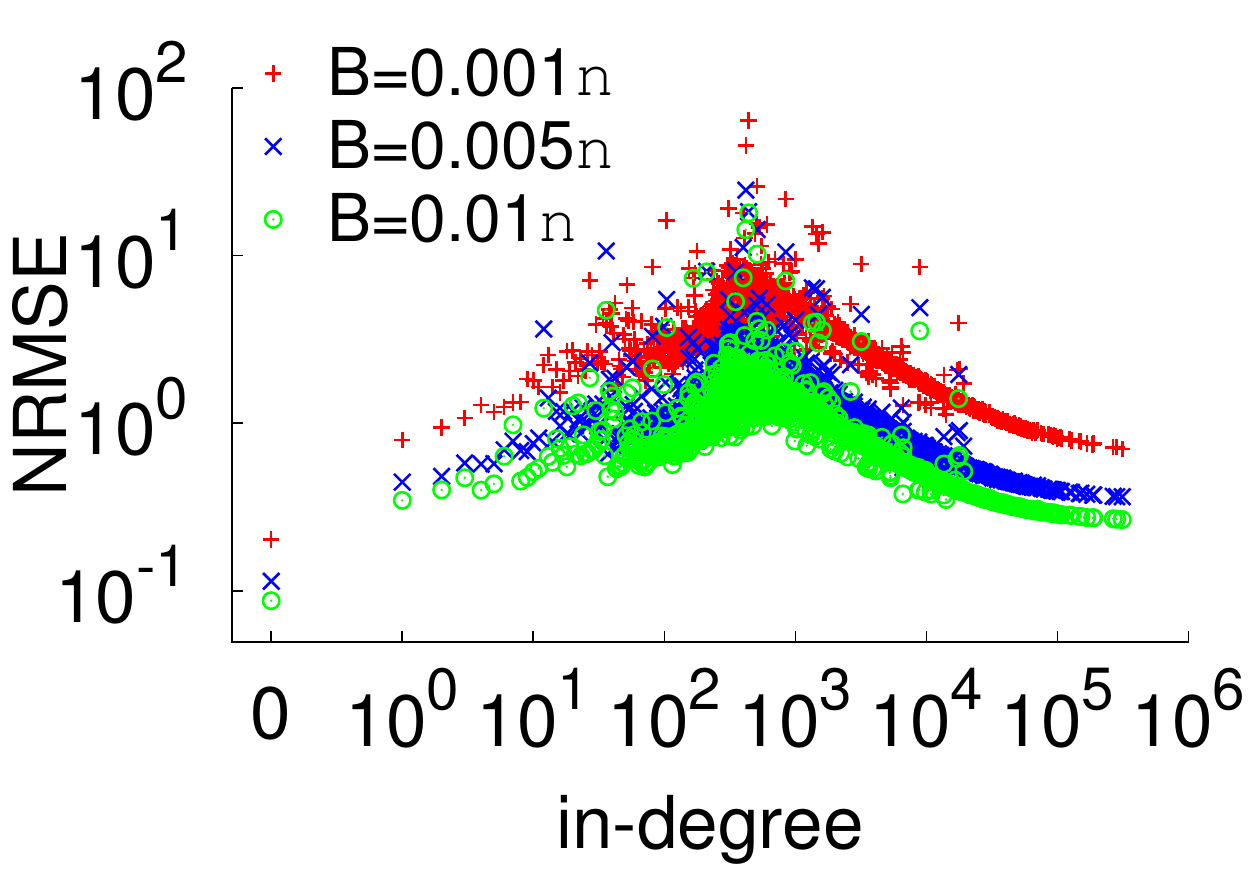}} 
\subfloat[Out-degree NRMSE]{
	\includegraphics[width=.5\linewidth]{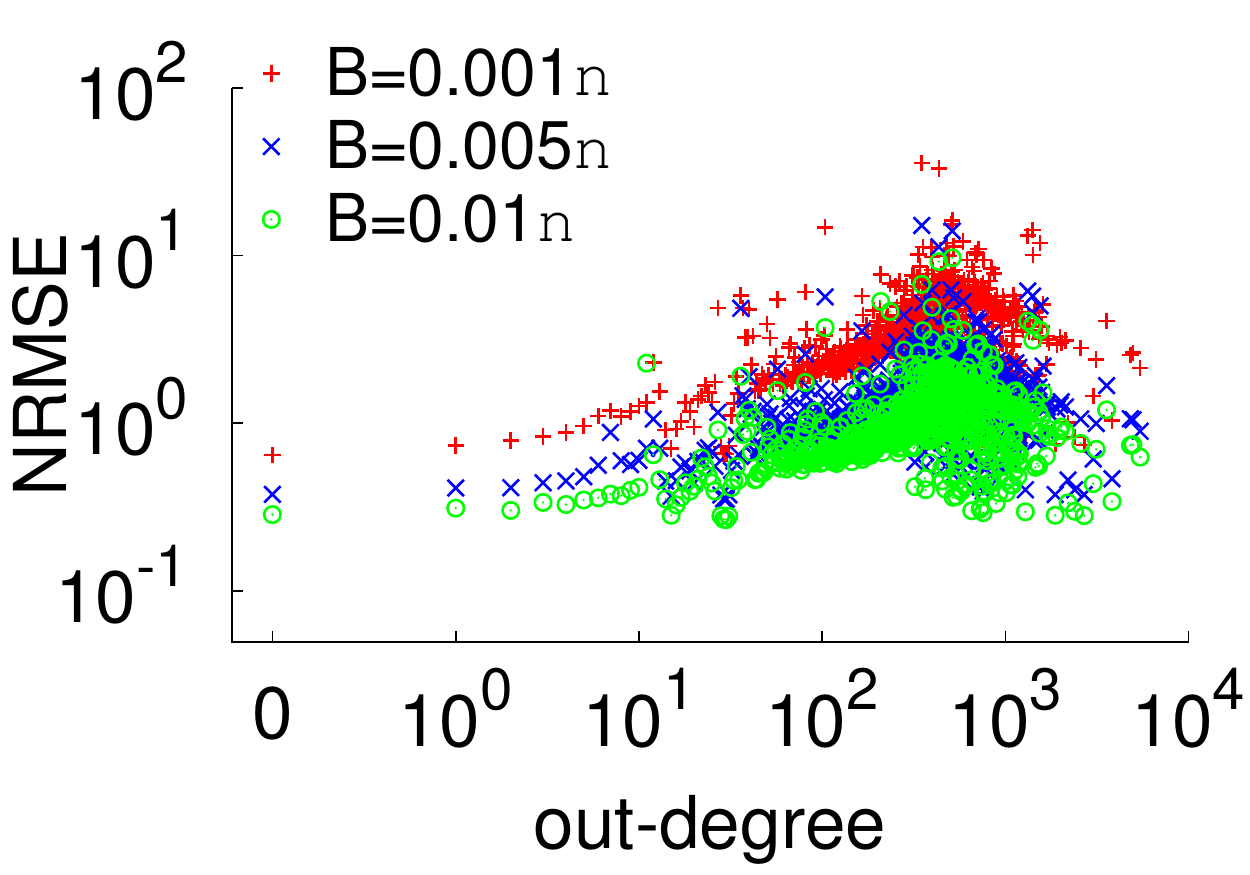}} \\
\subfloat[In-degree NRMSE]{
	\includegraphics[width=.5\linewidth]{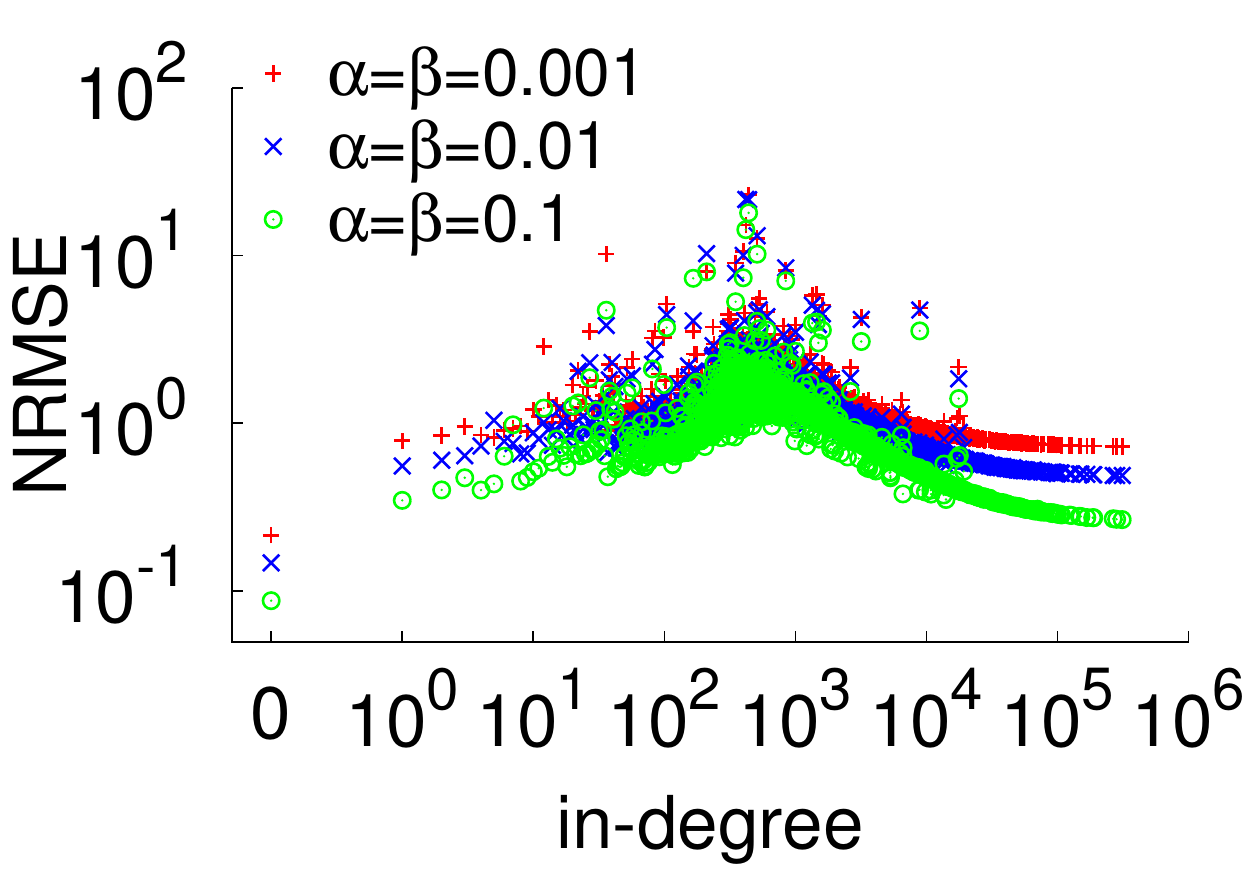}} 
\subfloat[Out-degree NRMSE]{
	\includegraphics[width=.5\linewidth]{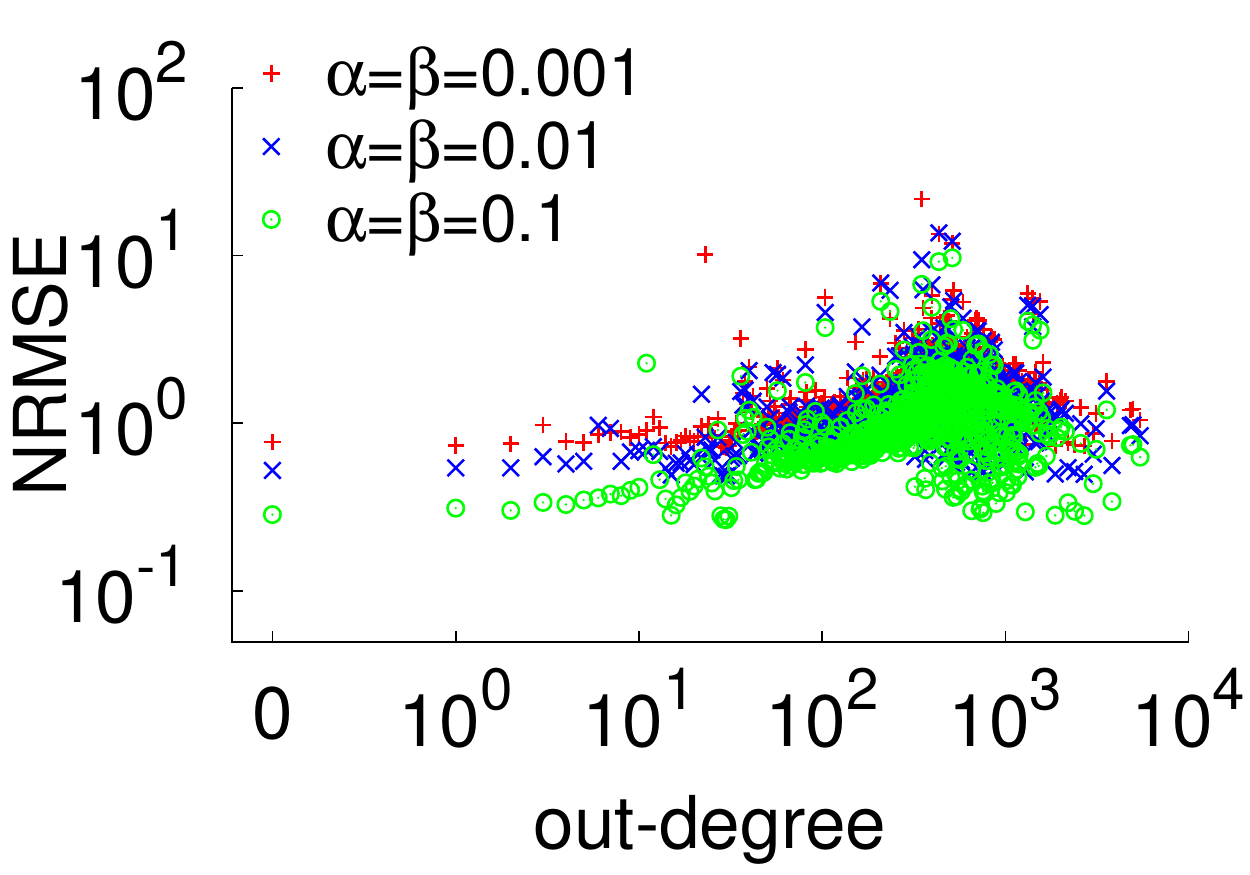}}
\caption{\RW{T}\RW{A} estimates and NRMSE. (We use $\alpha=\beta=0.1$ in
(a)-(d), and $B=0.01n$ in (e)-(f). Each result is averaged over $10,000$
runs.)}
\label{fig:mt_rwtrwa}
\end{figure}

\section{\textbf{Related Work}} \label{sec:related}

We review the related literature in this section.

Sampling methods, especially random walk based sampling methods, have been
widely used to characterize complex networks.
These applications include, but are not limited to, estimating peer
statistics in peer-to-peer networks~\cite{Gkantsidis2006,Massoulie2006},
uniformly sampling users from online social
networks~\cite{Gjoka2010,Gjoka2011}, characterizing structure properties of
large networks~\cite{Katzir2011,Hardiman2013,Seshadhri2013,Wang2014a}, and
measuring statistics of point-of-interests within an area on
maps~\cite{Wang2014} or venues in a region on LBSNs~\cite{Li2012,Li2014}.
The above literature is mostly concerned with sampling methods that seek to
\emph{directly} sample nodes (or samples) in target graphs (or sample
spaces).
However, direct sampling is not always efficient as we argued in this work.

When target graph (or sample space) can not be directly sampled or
direct sampling is inefficient, several methods based on graph manipulation
are proposed to improve sampling efficiency. 
For example, Gjoka et al.~\cite{Gjoka2011a} use different kinds
of relations (i.e., edges) to build a \emph{multigraph}, which is better
connected than any individual graph formed by only one kind of relations.
Therefore the random walk can converge fast on this multigraph.
Zhou et al.~\cite{Zhou2013} exploit several criteria to rewire target graph
on-the-fly so as to increase conductance and reduce mixing time of random
walks.
Our method differs from theirs that we do not manipulate target graphs.
We study a method on how to utilize auxiliary graph and affiliation graph
to assist sampling on target graph indirectly.

Birnbaum and Sirken~\cite{Birnbaum1965} designed a survey method for
estimating the number of diagnosed cases of a rare disease in a population.
Directly sampling patients of a rare disease is obviously inefficient so
they studied how to sample hospitals.
Their method motivates us to design the \VS{A} method. However, as we
pointed out, \VS{A} method cannot sample nodes that are not connected to
the auxiliary graph, and we overcome this problem by designing \RW{T}\VS{A}
and \RW{T}\RW{A}.
Our work also complements existing sampling methods such as random walk
with jumps~\cite{Avrachenkov2010,Ribeiro2012b} and Frontier
sampling~\cite{Ribeiro2010} by removing the necessity of uniform vertex
sampling on target graph.

\section{\textbf{Conclusion}} \label{sec:conclusion}

In this work we designed three sampling methods on a hybrid
social-affiliation network.
The concept of hybrid social-affiliation network can help sampling a graph
indirectly but efficiently.
The reason of effectiveness behind our methods lies in the improvement of
connectedness of target graph with the assistances of the other two
graphs.
We demonstrated the effectiveness of these sampling methods on both
synthetic and real datasets.
Our method complements existing methods in the area of graph sampling.

\bibliographystyle{abbrv}

\begin{thebibliography}{10}

\bibitem{WeiboPlace}
Weibo place.
\newblock \url{http://place.weibo.com}, March 2014.

\bibitem{WeiboLimit}
Weibo rate limit.
\newblock \url{http://goo.gl/WlohOj}, March 2014.

\bibitem{Mtime}
Mtime.
\newblock \url{http://www.mtime.com}, March 2014.

\bibitem{Aral2012a}
S.~Aral and D.~Walker.
\newblock Identifying influential and susceptible members of social networks.
\newblock {\em Science}, 337:337--341, 2012.

\bibitem{Asur2010}
S.~Asur and B.~A. Huberman.
\newblock Predicting the future with social media.
\newblock In {\em WI-IAT}, 2010.

\bibitem{Avrachenkov2010}
K.~Avrachenkov, B.~Ribeiro, and D.~Towsley.
\newblock Improving random walk estimation accuracy with uniform restarts.
\newblock In {\em the 7th Workshop on Algorithms and Models for the Web Graph},
  2010.

\bibitem{Barabasi1999}
A.~L. Barab\'{a}si and R.~Albert.
\newblock Emergence of scaling in random networks.
\newblock {\em Science}, 286(5439):509--512, 1999.

\bibitem{Birnbaum1965}
Z.~W. Birnbaum and M.~G. Sirken.
\newblock Design of sample surveys to estimate the prevalence of rare diseases:
  Three unbiased estimates.
\newblock {\em Vital and Health Statistics}, 2(11):1--8, 1965.

\bibitem{Bollen2011}
J.~Bollen, H.~Mao, and X.-J. Zeng.
\newblock Twitter mood predicts the stock market.
\newblock {\em Journal of Computational Science}, 2(1):1--8, 2011.

\bibitem{Cho2011}
E.~Cho, S.~A. Myers, and J.~Leskovec.
\newblock Friendship and mobility: User movement in location-based social
  networks.
\newblock In {\em KDD}, 2011.

\bibitem{Gjoka2011a}
M.~Gjoka, C.~T. Butts, M.~Kurant, and A.~Markopoulou.
\newblock Multigraph sampling of online social networks.
\newblock {\em JSAC}, 29(9):1893--1905, 2011.

\bibitem{Gjoka2010}
M.~Gjoka, M.~Kurant, C.~T. Butts, and A.~Markopoulou.
\newblock Walking in {F}acebook: A case study of unbiased sampling of {OSN}s.
\newblock In {\em INFOCOM}, 2010.

\bibitem{Gjoka2011}
M.~Gjoka, M.~Kurant, C.~T. Butts, and A.~Markopoulou.
\newblock Practical recommendations on crawling online social networks.
\newblock {\em JSAC}, 29(9):1872--1892, 2011.

\bibitem{Gkantsidis2006}
C.~Gkantsidis, M.~Mihail, and A.~Saberi.
\newblock Random walks in peer-to-peer networks: Algorithms and evaluation.
\newblock {\em Performance Evaluation}, 63(3):241--263, March 2006.

\bibitem{Hardiman2013}
S.~J. Hardiman and L.~Katzir.
\newblock Estimating clustering coefficients and size of social networks via
  random walk.
\newblock In {\em WWW}, 2013.

\bibitem{Katzir2011}
L.~Katzir, E.~Liberty, and O.~Somekh.
\newblock Estimating sizes of social networks via biased sampling.
\newblock In {\em WWW}, 2011.

\bibitem{Lazer2009}
D.~Lazer, A.~Pentland, L.~Adamic, S.~Aral, A.-L. Barabasi, D.~Brewer,
  N.~Christakis, N.~Contractor, J.~Fowler, M.~Gutmann, T.~Jebara, G.~King,
  M.~Macy, D.~Roy, and M.~Van~Alstyne.
\newblock Computational social science.
\newblock {\em Science}, 323:721--723, 2009.

\bibitem{Li2012}
Y.~Li, M.~Steiner, L.~Wang, Z.-L. Zhang, and J.~Bao.
\newblock Dissecting {F}oursquare venue popularity via random region sampling.
\newblock In {\em CoNEXT}, 2012.

\bibitem{Li2014}
Y.~Li, L.~Wang, M.~Steiner, J.~Bao, and T.~Zhu.
\newblock Region sampling and estimation of geosocial data with dynamic range
  calibration.
\newblock In {\em ICDE}, 2014.

\bibitem{Lovasz1993}
L.~Lov\'{a}sz.
\newblock Random walks on graphs: A survey.
\newblock {\em Combinatorics, Paul Erd\"{o}s is Eighty}, 2:353--397, 1993.

\bibitem{Massoulie2006}
L.~Massouli\'{e}, E.~L. Merrer, A.-M. Kermarrec, and A.~Ganesh.
\newblock Peer counting and sampling in overlay networks: Random walk methods.
\newblock In {\em PODC}, 2006.

\bibitem{Meyn2009}
S.~Meyn and R.~L. Tweedie.
\newblock {\em Markov Chains and Statistic Stability}.
\newblock Cambridge University Press, second edition, 2009.

\bibitem{Mislove2007}
A.~Mislove, M.~Marcon, K.~P. Gummadi, P.~Druschel, and B.~Bhattacharjee.
\newblock Measurement and analysis of online social networks.
\newblock In {\em IMC}, 2007.

\bibitem{Mohaisen2010}
A.~Mohaisen, A.~Yun, and Y.~Kim.
\newblock Measuring the mixing time of social graphs.
\newblock In {\em IMC}, 2010.

\bibitem{Mondal2012}
M.~Mondal, B.~Viswanath, P.~Druschel, K.~P. Gummadi, A.~Clement, A.~Mislove,
  and A.~Post.
\newblock Defending against large-scale crawls in online social networks.
\newblock In {\em CoNEXT}, 2012.

\bibitem{Newman2003}
M.~E.~J. Newman.
\newblock The structure and function of complex networks.
\newblock {\em SIAM Review}, 45(2):167--256, 2003.

\bibitem{Ribeiro2010}
B.~Ribeiro and D.~Towsley.
\newblock Estimating and sampling graphs with multidimensional random walks.
\newblock In {\em IMC}, 2010.

\bibitem{Ribeiro2012b}
B.~Ribeiro, P.~Wang, F.~Murai, and D.~Towsley.
\newblock Sampling directed graphs with random walks.
\newblock In {\em INFOCOM}, 2012.

\bibitem{Robert2004}
C.~P. Robert and G.~Casella.
\newblock {\em Monte Carlo Statistic Methods}.
\newblock Springer, second edition, 2004.

\bibitem{Seshadhri2013}
C.~Seshadhri, A.~Pinar, and T.~G. Kolda.
\newblock Triadic measures on graphs: The power of wedge sampling.
\newblock In {\em SDM'13}, 2013.

\bibitem{Wang2014}
P.~Wang, W.~He, and X.~Liu.
\newblock An efficient sampling method for characterizing points of interests
  on maps.
\newblock In {\em ICDE}, 2014.

\bibitem{Wang2014a}
P.~Wang, J.~C. Lui, B.~Ribeiro, D.~Towsley, J.~Zhao, and X.~Guan.
\newblock Efficiently estimating motif statistics of large networks.
\newblock {\em TKDD}, 2014.

\bibitem{Wasserman1994}
S.~Wasserman and K.~Faust.
\newblock {\em Social Network Analysis: Methods and Applications}.
\newblock Cambridge University Press, 1994.

\bibitem{Watts2004}
D.~J. Watts.
\newblock The new science of networks.
\newblock {\em Annual Review of Sociology}, 30(1):243--270, 2004.

\bibitem{Zhou2013}
Z.~Zhou, N.~Zhang, Z.~Gong, and G.~Das.
\newblock Faster random walks by rewiring online social networks on-the-fly.
\newblock In {\em ICDE}, 2013.

\end{thebibliography}

\end{document}